\documentclass[journal,onecolumn]{IEEEtran}
\usepackage{amsmath,amssymb,amsthm,mathtools}
\usepackage{enumitem}
\usepackage{cite}
\usepackage{booktabs}
\usepackage{array}
\usepackage{longtable}
\usepackage{xcolor}
\usepackage{caption}
\usepackage[hypertexnames=false,hidelinks]{hyperref}
\newtheorem{theorem}{Theorem}

\newtheorem{lemma}{Lemma}
\newtheorem{corollary}{Corollary}
\newtheorem{definition}{Definition}

\newcommand{\F}{\mathbb F}

\newcommand{\Tr}{\operatorname{Tr}}
\newcommand{\Span}{\operatorname{Span}}

\newcommand{\wt}{\operatorname{wt}}

\title{Minimum distances of primitive narrow-sense BCH codes via good zero-sets}
\author{Run Zheng\thanks{Run Zheng is with the Department of Mathematics, The Hong Kong
University of Science and Technology, Hong Kong, China (e-mail: zhengrun@ust.hk).}}
\begin{document}
\maketitle

\begin{abstract}
Determining the exact minimum distances of BCH codes remains a
open problem. We establish the minimum distances of several
families of primitive narrow-sense BCH codes, showing that they
attain their designed distances. Our approach centers on
$\mathbb{F}_q$-good zero-sets, which we introduce through a derivative
condition on their vanishing polynomials. We show that a $q$-ary
primitive narrow-sense BCH code of length $q^m-1$ and designed
distance $2\leq\delta\leq q^m-1$ has minimum distance $\delta$ if and
only if there exists an $\mathbb{F}_q$-good zero-set of cardinality
$\delta+1$ in the finite field
$\mathbb{F}_{q^m}$ with $q^m$ elements. To construct $\mathbb{F}_q$-good zero-sets, we
develop several methods based on polynomial substitutions, power maps,
and shifted inverses, as well as direct constructions using
polynomials of special forms.
Together with suitable initial $\mathbb{F}_q$-good zero-sets, including those
arising from known minimum-distance results, these methods yield
new good zero-sets of various cardinalities and hence families of primitive narrow-sense BCH codes whose
minimum distances equal their designed distances. These families
cover a broad range of designed distances, with several known
minimum-distance results recovered as special cases.
\end{abstract}


\section{Introduction}
BCH codes were introduced independently by Hocquenghem
\cite{Hocquenghem1959} and by Bose and Ray-Chaudhuri
\cite{BoseRayChaudhuri1960}. Their rich algebraic structure,
efficient decoding algorithms, and strong error-correcting capability
have led to extensive study and widespread use in communication and
data storage.
Despite  decades of research, our understanding of these codes remains surprisingly limited. Even for \textit{primitive  narrow-sense}   BCH codes, the most extensively studied subclass, fundamental parameters  have been determined in only a few specific cases. In particular, determining the exact minimum distance remains a challenging open problem, as highlighted in the surveys by Charpin \cite{charpin1998open} and Ding and Li \cite{DingLi2024}.

Throughout  this paper, let $q=p^e$ for some prime $p$ and positive integer $e$, let $m$ be a positive integer, and let $n=q^m-1$. Denote by 
$\mathcal{C}_{(q,m,\delta)}$ the $q$-ary primitive  narrow-sense BCH code of length $n$ and designed distance $\delta$, where $2\leq \delta\leq n$. By  the BCH bound \cite{Hocquenghem1959,BoseRayChaudhuri1960}, the minimum distance of $\mathcal{C}_{(q,m,\delta)}$, denoted by $d\left(\mathcal{C}_{(q,m,\delta)}\right)$, satisfies $d\left(\mathcal{C}_{(q,m,\delta)}\right)\geq \delta$. One subtlety is that 
 different designed distances may define the same BCH code;
that is, $\mathcal{C}_{(q,m,\delta)}=\mathcal{C}_{(q,m,\delta')}$ may hold
for distinct $\delta$ and $\delta'$.
The largest $\delta'$ defining $\mathcal{C}_{(q,m,\delta)}$ is called
its \emph{Bose distance}, denoted by
$d_B\bigl(\mathcal{C}_{(q,m,\delta)}\bigr)$.
Consequently,  we have 
\begin{equation*}
d(\mathcal{C}_{(q,m,\delta)})
\geq d_B(\mathcal{C}_{(q,m,\delta)})
\geq\delta.
\end{equation*}
Determining the Bose distance therefore provides a sharper lower bound on the minimum distance and has attracted considerable attention
\cite{CherchemEtAl2020,D2015,DFZ2017,YF2000,YH1996,zheng2025,zheng2026}.
Beyond the determination of the Bose distance, a fundamental problem is to identify when the minimum distance equals the designed distance or the Bose distance. This  was singled out as a research problem in the classical monograph of MacWilliams and Sloane \cite[Ch. 9, Research Problem 9.1]{MacWilliamsSloane}.

Along this line of research, the minimum distance has been shown to
attain the designed distance, and consequently the Bose distance,
for several families of primitive narrow-sense BCH codes.
A classical result states that
$d(\mathcal{C}_{(q,m,\delta)})=\delta$ whenever $\delta$ divides $n=q^m-1$;
see \cite[p.~247]{BBKK2006}. Peterson \cite{Peterson1969} proved that
a primitive narrow-sense BCH code with designed distance $q^t-1$
has minimum distance $q^t-1$ for $1\leq t\leq m-1$. Berlekamp \cite{Berlekamp1970} proved that the designed
distance is attained for certain families of binary primitive BCH
codes by exploiting the weight distributions of subcodes of
second-order Reed--Muller codes, and Kasami and Lin
\cite{KasamiLin1972} extended this approach to further binary families.
Additional families were obtained through connections with punctured
generalized Reed--Muller codes \cite{sun2026,D2015}, trace
representations together with quadratic-form techniques 
\cite{DFZ2017}, and association schemes of quadratic and bilinear
forms \cite{Li2017}. More recently, Tiwari and Kewat
\cite{TiwariKewat2026} used polynomial interpolation to construct
codewords of weight equal to the designed distance. Chen et al.
\cite{ChenEtAlBCH2026,Chen2026} used locator polynomials and connections
with Goppa codes to obtain further families attaining their designed
distances.

In this paper, we provide a criterion for $\mathcal{C}_{(q,m,\delta)}$ to attain its designed distance $\delta$ by introducing the concept of an $\mathbb{F}_q$-good zero-set as follows.   
For a finite field $\mathbb{F}$, we denote by
$\mathbb{F}^*=\mathbb{F}\setminus\{0\}$ its multiplicative group.
\begin{definition}\label{definition1}
Let $q$ be a prime power and $\mathbb{F}_q$ be the finite field of $q$ elements, and 
let $\mathbb{F}$ be a finite extension  of $\mathbb{F}_q$.
A finite subset $\mathcal{A}$ of $\mathbb{F}$ is called an
\emph{$\mathbb{F}_q$-good zero-set} if $0\in\mathcal{A}$ and there exists
$\lambda\in\mathbb{F}^*$ such that
\[
P_{\mathcal{A}}'(a)\in\lambda\mathbb{F}_q^*
\quad\text{for all }a\in\mathcal{A},
\]
where $P_{\mathcal{A}}(x)=\prod_{\alpha\in\mathcal{A}}(x-\alpha)$ is
the vanishing polynomial of $\mathcal{A}$ and $P_{\mathcal{A}}'$ denotes
its formal derivative. In particular, an $\mathbb{F}_q$-good zero-set is called
\emph{strict} if one can take $\lambda=1$.
\end{definition}

Our  criterion shows that
$
d\bigl(\mathcal{C}_{(q,m,\delta)}\bigr)=\delta
$ 
if and only if there exists an $\mathbb{F}_q$-good zero-set
$\mathcal{A}\subseteq\mathbb{F}_{q^m}$ of cardinality  $\delta+1$.
This characterization is closely related to the locator-polynomial criterion of Chen et al. \cite{ChenEtAlBCH2026}.
In terms of
$\mathbb{F}_q$-good zero-sets, their criterion is equivalent to the existence of such a set
$\mathcal{A}\subseteq\mathbb{F}_{q^m}$ of cardinality $\delta+1$
containing $1$.
Since any codeword of weight $\delta$ can be cyclically
shifted so that one of its locators is $1$, the two criteria are
equivalent as existence statements. Working without the requirement
$1\in\mathcal{A}$ nevertheless provides greater flexibility in
constructing good zero-sets of the required cardinality.
Moreover, the derivative formulation is naturally compatible with
polynomial composition via the chain rule. It allows good zero-sets
to be constructed from known ones by taking suitable inverse images
and also leads to direct constructions using suitable polynomials
that split completely over the ambient field.

Based on this observation, we develop several methods for constructing
$\mathbb{F}_q$-good zero-sets. The first two share a common idea.
Given an $\mathbb{F}_q$-good zero-set $\mathcal{A}$, we choose a subset
$\mathcal{S}\subseteq\mathcal{A}$ and a polynomial $L(x)$ so that
\[
\prod_{\alpha\in\mathcal{A}\setminus\mathcal{S}}(x-\alpha)
\prod_{\alpha\in\mathcal{S}}(L(x)-\alpha)
\]
is the vanishing polynomial of a larger $\mathbb{F}_q$-good zero-set.
The $p$-linearized polynomial substitution takes
$\mathcal{S}=\mathcal{A}$ and a $p$-linearized polynomial $L(x)=x^p-u^{p-1}x$; equivalently, it substitutes
$L(x)=x^p-u^{p-1}x$ for $x$ in the vanishing polynomial
$P_{\mathcal{A}}(x)$.
The $v$-th-power construction instead takes
$\mathcal{S}=\mathcal{A}\setminus \{0\}$ and $L(x)=x^v$ for a suitable
integer $v$.
 Under suitable conditions, the shifted inverse
construction applies the map $\alpha\mapsto(\alpha-\beta)^{-1}$
and adjoins $0$, producing an $\mathbb{F}_q$-good zero-set from a strict
$\mathbb{F}_q$-good zero-set. Finally, we give direct constructions
using polynomials of special forms.

Applying these methods, we construct $\mathbb{F}_q$-good zero-sets of
various cardinalities and thereby obtain several families of primitive
narrow-sense BCH codes that attain their designed distances.
We start with a  $t$-element subset of $\mathbb{F}_q$ containing
$0$ and combine iterated $p$-linearized polynomial substitutions with
the $v$-th-power construction. This yields the family with designed
distance $\delta=\bigl(1+v(tp^u-1)\bigr)p^b-1$ under the parameter
conditions in Theorem \ref{theorem6}. Starting instead with an
$\mathbb{F}_q$-good zero-set of cardinality $q^t+2$, obtained from
\cite[Theorem~8]{Chen2026}, we apply the same procedure to obtain
the family in Theorem \ref{theorem7}.
We also apply the shifted inverse construction to suitable strict
$\mathbb{F}_q$-good zero-sets to obtain the family in
Theorem \ref{theorem8}. 
Finally, polynomials of special forms yield codes with designed
distances $q^{2r}+1$ when $5r\mid m$ and $q^{3r}+1$ when $8r\mid m$, as established in
Theorem \ref{theorem9}. These results verify two additional cases of
\cite[Conjecture~2]{DingDuZhou2015}, which predicts equality of the
minimum and Bose distances of $\mathcal{C}_{(q,m,q^t+1)}$.
Taken together, these constructions cover a broad range of designed
distances and recover several known results as special cases.
To complement these results, we tabulate selected examples from these
families that are optimal or attain the best-known minimum distance
recorded in the table of the best-known linear codes maintained by Markus Grassl at http://www.codetables.de, which is called \textit{Database} later in this paper. 

The rest of this paper is organized as follows. Section II presents
the preliminaries on primitive narrow-sense BCH codes. Section~III proves the
good zero-set characterization in Theorem\ref{theorem1}.
Section IV develops the constructions of $\mathbb{F}_q$-good
zero-sets. Section V applies these constructions to determine the
exact minimum distances of several families of primitive narrow-sense
BCH codes and presents selected numerical examples. Section VI
concludes the paper. The appendix contains proofs of auxiliary lemmas.

\section{Preliminaries}
 Denote by $\mathbb{F}_q$ the finite field with $q$ elements, and by $\mathbb{F}_q^n$ the $n$-dimensional vector space over $\mathbb{F}_q$.
A $q$-ary linear code $\mathcal{C}$ of length $n$ is a linear subspace
of the $n$-dimensional vector space $\mathbb{F}_q^n$. The minimum distance of a nonzero linear code
$\mathcal{C}\subseteq\mathbb{F}_q^n$ is
$$d(\mathcal{C})=\min\{\wt(\mathbf{c}): \mathbf{c}\in\mathcal{C}, \mathbf{c}\neq \mathbf{0}\},$$
where $\wt(\mathbf{c})$ denotes the Hamming weight of $\mathbf{c}$,
namely the number of its nonzero coordinates. 
A linear code $\mathcal{C}$ is called \emph{cyclic} if
$(c_0,c_1,\ldots,c_{n-1})\in\mathcal{C}$ implies that
$(c_{n-1},c_0,\ldots,c_{n-2})\in\mathcal{C}$.
Identify each vector $(c_0,c_1,\ldots,c_{n-1})\in\mathbb{F}_q^n$
with its polynomial representation
\[
c_0+c_1x+\cdots+c_{n-1}x^{n-1}\in\mathbb{F}_q[x]/(x^n-1).
\]
Then a linear code is cyclic if and only if it is an ideal of the
quotient ring $\mathbb{F}_q[x]/(x^n-1)$. Since this quotient  is a
principal ideal ring, every cyclic code $\mathcal{C}$ has a unique
monic generator $g(x)$ dividing $x^n-1$, and we write
$\mathcal{C}=\langle g(x)\rangle$. The polynomial $g(x)$ is called
the \emph{generator polynomial} of $\mathcal{C}$, and
$h(x)=(x^n-1)/g(x)$ is called its \emph{parity-check polynomial}.

Let $\alpha\in\mathbb{F}_{q^m}$ be a primitive element. 
For each integer $i$ with $0\leq i\leq n-1$,  we denote by $m_i(x)$ the minimal polynomial of $\alpha^i$ over $\mathbb{F}_q$. 
The $q$-ary primitive narrow-sense
BCH code of length $n$ with designed distance $\delta $  is the cyclic code  with the generator polynomial 
\begin{equation*}
    g_{(q,m,\delta)}(x)=\mathrm{lcm}\left(m_{1}(x), m_{2}(x),\ldots,m_{\delta-1}(x)\right), 
    \end{equation*}
where $2\leq \delta\leq n$ and  $\mathrm{lcm}$ denotes the least common multiple of the polynomials.    

Throughout the remainder of this paper, we index the coordinates of
vectors in $\mathbb{F}_q^n$ by the elements of $\mathbb{F}_{q^m}^*$
in the order $1,\alpha,\ldots,\alpha^{n-1}$.
Accordingly, we write each vector $\mathbf{c}\in\mathbb{F}_q^n$ as
$\mathbf{c}=(c_\beta)_{\beta\in\mathbb{F}_{q^m}^*}
=(c_1,c_\alpha,\ldots,c_{\alpha^{n-1}})$. Its polynomial representation then takes the form
$c(x)=\sum_{j=0}^{n-1}c_{\alpha^j}x^j$.
By the definition of $g_{(q,m,\delta)}(x)$, the vector $\mathbf{c}$
belongs to $\mathcal{C}_{(q,m,\delta)}$ if and only if
$c(\alpha^i)=0$ for all $1\leq i\leq\delta-1$.
It follows that
\begin{equation}\label{verify}
\mathcal{C}_{(q,m,\delta)}
=
\left\{
\mathbf{c}\in\mathbb{F}_q^n:
\sum_{\beta\in\mathbb{F}_{q^m}^*}c_\beta\beta^i=0
\text{ for all }1\leq i\leq\delta-1
\right\}.
\end{equation}

\section{Good zero-sets}
We now establish a necessary and sufficient condition for
$\mathcal{C}_{(q,m,\delta)}$ to attain its designed distance,
using the $\mathbb{F}_q$-good zero-sets introduced in
Definition~\ref{definition1}.
Recall that
$P_{\mathcal{A}}(x)=\prod_{\beta\in\mathcal{A}}(x-\beta)$
denotes the vanishing polynomial of a finite set $\mathcal A$. 
  We write    $|\mathcal A|$ for its cardinality and $\deg f$ for the degree of a nonzero polynomial $f$. 
We first give the  following Lemma, whose proof is  provided
in the appendix.

\begin{lemma}\label{adlemma1}
Let $\mathbb{F}$ be a field, and let
$\mathcal{D}\subseteq \mathbb{F}$ be a finite set 
with $|\mathcal{D}|\geq 2$. 
Then 
\begin{equation}\label{lemma1e1}
        \sum_{\beta\in \mathcal{D}}\frac{\beta^j}{P_{\mathcal{D}}'(\beta)}=0\qquad \hbox{for all }0\le j\le |\mathcal{D}|-2.
        \end{equation}
\end{lemma}

\begin{theorem}\label{theorem1}
Let $\delta$ be an integer satisfying $2 \le \delta \le q^m - 1$. Then $d(\mathcal{C}_{(q,m,\delta)}) = \delta$  if and only if there exists an $\mathbb{F}_q$-good zero-set $\mathcal{A} \subseteq \mathbb{F}_{q^m}$ such that $|\mathcal{A}| = \delta + 1$.
\end{theorem}
\begin{proof} 
\textit{Sufficiency.} Suppose that $\mathcal{A}\subseteq \mathbb{F}_{q^m}$ is an $\mathbb{F}_q$-good zero-set with $|\mathcal{A}|=\delta+1$.
By definition, there exists $\lambda\in \mathbb{F}_{q^m}^*$ such that
\begin{equation}\notag     P_{\mathcal{A}}'(\beta)\in\lambda\F_q^*
        \quad \hbox{for all }\beta\in \mathcal{A}.
\end{equation}
Let $\mathcal{S}=\mathcal{A}\setminus\{0\}$. 
It is easy to see that   \begin{equation}\notag
P_{\mathcal{A}}'(\beta)= \beta
P_{\mathcal{S}}^{\prime}(\beta)
\quad \hbox{for all } \beta\in \mathcal{S}.
\end{equation} 
It follows that 
\begin{equation}\notag
   \frac{\lambda}{\beta P_{\mathcal{S}}'(\beta)}\in\mathbb F_q^*
        \quad
       \text{for all    }\beta\in \mathcal{S}.
\end{equation}
Define a vector $\mathbf{c}=(c_\beta)_{\beta\in\mathbb{F}_{q^m}^*}\in\mathbb{F}_q^n$ by
\begin{equation}\notag  c_\beta=\begin{cases}
       0& \hbox{if }\beta\not \in  \mathcal{S},\\
   \frac{\lambda}{\beta P_{\mathcal{S}}'(\beta)}& \hbox{if }\beta\in \mathcal{S}.  
   \end{cases} 
\end{equation}
It is clear that $\wt(\mathbf c)=|\mathcal{S}|=\delta$.  Applying Lemma~\ref{adlemma1} with $(\mathcal{D},j)=(\mathcal{S},i-1)$, we obtain 
\begin{equation}\notag
\begin{aligned}
        \sum_{\beta\in \mathbb{F}_{q^m}^*}c_\beta\beta^i
        = \sum_{\beta\in \mathcal{S}}
\left(\frac{ \lambda }{\beta P_{\mathcal{S}}'(\beta)}\right)\beta^{i}
      = \sum_{\beta\in \mathcal{S}}
\frac{ \lambda \beta^{i-1} }{P_{\mathcal{S}}'(\beta)}
        =
        0 \quad \hbox{ for all }1\le i\leq \delta-1. 
\end{aligned}
\end{equation}
By \eqref{verify}, the vector $\mathbf{c}$ is a codeword of 
$\mathcal{C}_{(q,m,\delta)}$. The BCH bound now gives  $d(\mathcal{C}_{(q,m,\delta)}) = \delta$. 

\textit{Necessity.}
Suppose that $d(\mathcal{C}_{(q,m,\delta)})=\delta$, and let
$(c_{\gamma})_{\gamma\in \mathbb{F}_{q^m}^*}\!\!\in\mathcal{C}_{(q,m,\delta)}$ be a codeword of
weight $\delta$.
Let 
$\mathcal{S}=\{\gamma\in\mathbb{F}_{q^m}^*:c_\gamma\neq 0\}$
and $\mathcal{A}=\mathcal{S}\cup\{0\}$.
Then $|\mathcal{S}|=\delta$ and $|\mathcal{A}|=\delta+1$.
By \eqref{verify}, we obtain 
\begin{equation*}
\sum_{\gamma\in\mathbb{F}_{q^m}^*}c_\gamma\gamma^i= \sum_{\gamma\in\mathcal{S}}c_\gamma\gamma^i=0
\quad
\text{for all }1\leq i\leq\delta-1.
\end{equation*}
By linearity, it follows that
\begin{equation}\label{lsb}
\sum_{\gamma\in\mathcal{S}}c_\gamma f(\gamma)
=
f(0)\sum_{\gamma\in\mathcal{S}}c_\gamma
\end{equation}
for any  polynomial $f\in\mathbb{F}_{q^m}[x]$
with  $\mathrm{deg}(f)\leq \delta-1$.

For each $\beta\in\mathcal{S}$, define the polynomial 
\begin{equation*}
f_\beta(x)
=
\prod_{\gamma\in\mathcal{S}\setminus\{\beta\}}(x-\gamma).
\end{equation*}
It is clear that  $\mathrm{deg}(f_\beta(x))=\delta-1$ and $f_\beta(\gamma)=0$ for any $\gamma\in \mathcal{S}\setminus\{\beta\}$.
Note that  $P_{\mathcal{A}}(x)=x(x-\beta)f_\beta(x)$. Thus, we have
\begin{equation*}
f_\beta(\beta)=\frac{P_{\mathcal{A}}'(\beta)}{\beta}
\quad\text{and}\quad
f_\beta(0)=-\frac{P_{\mathcal{A}}'(0)}{\beta}.
\end{equation*}
Applying the  identity \eqref{lsb} to $f_\beta$ gives
\begin{equation*}
c_\beta f_{\beta}(\beta)
=\sum_{\gamma\in\mathcal{S}}c_\gamma f_{\beta}(\gamma)= 
f_\beta(0)
\sum_{\gamma\in\mathcal{S}}c_\gamma
\quad
\text{for all }\beta\in\mathcal{S}.
\end{equation*} 
Consequently, we have 
$$c_{\beta}P'_{\mathcal{A}}(\beta)=-P_{\mathcal{A}}'(0)
\sum_{\gamma\in\mathcal{S}}c_\gamma \quad \text{for all }\beta\in\mathcal{S}.$$
Let
$\lambda=-P_{\mathcal{A}}'(0)
\sum_{\gamma\in\mathcal{S}}c_\gamma$.
Noting that  $c_\beta\neq 0$ and
$P_{\mathcal{A}}'(\beta)\neq 0$ for all
$\beta\in\mathcal{S}$,  we have 
\begin{equation*}
\frac{\lambda}{P_{\mathcal{A}}'(\beta)}
=
c_\beta\in\mathbb{F}_q^*
\qquad
\text{for all }\beta\in\mathcal{S}.
\end{equation*}
Moreover, since $P_{\mathcal{A}}'(0)\neq 0$, we have
\begin{equation*}
\frac{\lambda}{P_{\mathcal{A}}'(0)}
=
-\sum_{\gamma\in\mathcal{S}}c_\gamma
\in\mathbb{F}_q^*.
\end{equation*}
Therefore,
$P_{\mathcal{A}}'(\beta)\in\lambda\mathbb{F}_q^*$
for all $\beta\in\mathcal{A}$.
Hence $\mathcal{A}$ is an $\mathbb{F}_q$-good zero-set
of cardinality $\delta+1$. This completes the necessity part.
\end{proof}

\section{Constructions of \texorpdfstring{$\mathbb{F}_q$}{Fq}-good zero-sets}
In this section, we present several methods for constructing
$\mathbb{F}_q$-good zero-sets under suitable conditions.
We first introduce some notation for the constructions below.
Let $p$ be a prime, and let $\mathbb{F}$ be a finite extension
of $\mathbb{F}_p$.
For an $\mathbb{F}_p$-vector space $V$, we write
$\dim_{\mathbb{F}_p}(V)$ for its dimension.
For a subset $\mathcal{A}\subseteq\mathbb{F}$, we denote by
$\Span_{\mathbb{F}_p}(\mathcal{A})$ its $\mathbb{F}_p$-linear span,
consisting of all finite $\mathbb{F}_p$-linear combinations
of elements of $\mathcal{A}$.
The trace map
$\Tr_{\mathbb{F}/\mathbb{F}_p}\colon\mathbb{F}\to\mathbb{F}_p$
is defined by
\begin{equation*}
    \Tr_{\mathbb{F}/\mathbb{F}_p}(x)
    = \sum_{i=0}^{\ell-1}x^{p^i},
\end{equation*}
where $\ell=[\mathbb{F}:\mathbb{F}_p]$ is the extension degree.
For subsets $V,W\subseteq\mathbb{F}$ and a map $L\colon V\to W$,
we denote its kernel and image by
\[
\ker(L)=\{v\in V:L(v)=0\},\qquad
\operatorname{Im}(L)=\{L(v):v\in V\},
\]
respectively. For $S\subseteq V$, we write $L|_S$ for the
restriction of $L$ to $S$.

\subsection{\texorpdfstring{$p$}{p}-linearized polynomial substitution}

The following lemma is needed for our first construction. Its proof is given in the appendix.
\begin{lemma}\label{lemma2}
Let $q=p^e$ for some prime $p$ and positive integer $e$, and  
let $\mathbb{F}$ be a finite extension of $\mathbb{F}_p$.  Suppose that  
   $\alpha\in \mathbb{F}$ satisfies  $\Tr_{\mathbb{F}/\mathbb{F}_{p}}(\alpha)=0$.
   Then the polynomial $y^p-y-\alpha$ splits completely over $\mathbb{F}$. 
\end{lemma}

\begin{theorem}\label{theorem2}
Let $q=p^e$ for some prime $p$ and positive integer $e$, and let $\mathbb{F}$ be a finite extension  of $\mathbb{F}_q$. Suppose that  
$\mathcal{A}\subseteq\mathbb{F}$ is an $\mathbb{F}_q$-good zero-set satisfying  
$\Span_{\mathbb{F}_p}(\mathcal{A})\neq\mathbb{F}$.  Then there exists $u\in \mathbb{F}^*$ such that 
\begin{equation}\notag
    \mathcal{A}^{+} =\{x\in \mathbb{F}: L_u(x)\in \mathcal{A}\} 
\end{equation}  
is an $\mathbb{F}_q$-good zero-set, where $L_u(x)=x^{p}-u^{p-1}x$. Moreover, the set $\mathcal{A}^+$ satisfies 
\begin{equation}\notag
    |\mathcal{A}^{+}| = p|\mathcal{A}|  
\quad \text{and}\quad 
\dim_{\mathbb{F}_p} \Span_{\mathbb{F}_p}(\mathcal{A}^{+}) =\dim_{\mathbb{F}_p} \Span_{\mathbb{F}_p}(\mathcal{A}) + 1. 
\end{equation}
\end{theorem}
\begin{proof}
Since $\Span_{\mathbb{F}_p}(\mathcal{A})\neq\mathbb{F}$, there exists a nonzero element
$c\in\mathbb F$ such that
$\Tr_{\mathbb F/\mathbb F_p}(c\alpha)=0$ for all $\alpha\in \mathcal{A}$.
Moreover, noting that the map $x\mapsto x^p$ is surjective on
$\mathbb F$, there exists $u\in\mathbb F^*$ such that $u^{-p}=c$.
Then
\[
L_u(x)-\alpha=u^p(y^p-y-\alpha c),
\]
where $y=x/u$.
By Lemma~\ref{lemma2}, the polynomial  $L_u(x)-\alpha$ splits completely
over $\mathbb F$ for every $\alpha\in\mathcal A$.
Hence
\begin{equation}\notag
P_{\mathcal A}(L_u(x))
=\prod_{\alpha\in\mathcal A}\bigl(L_u(x)-\alpha\bigr)
\end{equation}
splits completely over $\mathbb F$, and $\mathcal A^+$
is the set of all roots of $P_{\mathcal A}(L_u(x))$.

We next show that $\mathcal{A}^{+}$ is an
$\mathbb{F}_q$-good zero-set.
Since $\mathcal{A}$ is an $\mathbb{F}_q$-good zero-set,
we have $0\in\mathcal{A}$, and there exists
$\lambda\in\mathbb{F}^{*}$ such that
\begin{equation}\notag
P_{\mathcal{A}}'(\alpha)\in\lambda\mathbb{F}_q^{*}
\qquad\text{for all }\alpha\in\mathcal{A}.
\end{equation}
Since $L_u(0)=0$, we also have $0\in\mathcal{A}^{+}$.
For every $x\in\mathcal{A}^{+}$, we have
$L_u(x)\in\mathcal{A}$, and hence
\begin{equation}\label{equation6}
\bigl(P_{\mathcal A}(L_u(x))\bigr)'
=-u^{p-1}P_{\mathcal{A}}'(L_u(x))
\in -u^{p-1}\lambda\mathbb{F}_q^{*}.
\end{equation}
In particular, every root of $P_{\mathcal{A}}(L_u(x))$ is simple.
Since this polynomial is monic and has root set $\mathcal{A}^{+}$,
it follows that
\begin{equation}\notag
P_{\mathcal{A}^{+}}(x)=P_{\mathcal{A}}(L_u(x)).
\end{equation}
Together with \eqref{equation6} and $0\in\mathcal{A}^{+}$,
this shows that $\mathcal{A}^{+}$ is an
$\mathbb{F}_q$-good zero-set.
Taking degrees in the above  identity  gives \begin{equation}\notag
\deg P_{\mathcal{A}^{+}}(x)
=p\deg P_{\mathcal{A}}(x).
\end{equation}
Therefore, we have $|\mathcal{A}^+|=p|\mathcal{A}|$.

Finally, we establish the desired relation between the dimensions of $\Span_{\mathbb{F}_p}(\mathcal{A}^{+})$ and $\Span_{\mathbb{F}_p}(\mathcal{A})$. 
By the choice of $u$ and the definition of $\mathcal{A}^{+}$,
we have $L_u(\mathcal{A}^{+})=\mathcal{A}$.
Note that $L_u$ is an $\mathbb{F}_p$-linear map. Thus, it follows that
\begin{equation}\notag
L_u(\operatorname{Span}_{\mathbb{F}_p}(\mathcal{A}^{+}))
=\operatorname{Span}_{\mathbb{F}_p}
  \bigl(L_u(\mathcal{A}^{+})\bigr)
=\operatorname{Span}_{\mathbb{F}_p}(\mathcal{A})
.
\end{equation}
Since $u\neq 0$,  
$L_u(x)=0$  if and only if 
$
    \left(\frac{x}{u}\right)^p - \left(\frac{x}{u}\right)=0,
$
which is further equivalent to 
${x}/{u}\in \mathbb{F}_p$. Therefore,  we have $\ker(L_u)=u\mathbb{F}_p$.
Moreover,  $0\in\mathcal{A}$
implies that $\ker(L_u)\subseteq\mathcal{A}^{+}\subseteq \operatorname{Span}_{\mathbb{F}_p}(\mathcal{A}^{+})$.
Thus
\begin{equation}\notag
\ker\bigl(L_u|_{\operatorname{Span}_{\mathbb{F}_p}(\mathcal{A}^{+})}\bigr)=\ker(L_u)=u\mathbb{F}_p.
\end{equation}
Applying the rank--nullity theorem to
$L_u|_{\operatorname{Span}_{\mathbb{F}_p}(\mathcal{A}^{+})}\colon \operatorname{Span}_{\mathbb{F}_p}(\mathcal{A}^{+})\to \operatorname{Span}_{\mathbb{F}_p}(\mathcal{A})$, we obtain
\begin{equation}\notag
\dim_{\mathbb{F}_p}\operatorname{Span}_{\mathbb{F}_p}(\mathcal{A}^{+})
=\dim_{\mathbb{F}_p}\operatorname{Span}_{\mathbb{F}_p}(\mathcal{A})+1.
\end{equation}
This completes the proof.
\end{proof}

\begin{corollary}\label{corollary2}
Let $q=p^e$ for some prime $p$ and positive integer $e$,
and let  $s$ be a positive integer. Suppose that $\mathcal{A}\subseteq \mathbb{F}_{q^s}$ is  an $\mathbb{F}_q$-good zero-set   satisfying  $\dim_{\mathbb{F}_p} \Span_{\mathbb{F}_p}(\mathcal{A})\leq r $ for some integer $r$.  Then   for any integer $u$ with $0\leq u\leq es-r$, there exists an $\mathbb{F}_q$-good zero-set $\mathcal{A}_{u}\subseteq \mathbb{F}_{q^s}$ such that  
\begin{equation}\label{equation7}
   |\mathcal{A}_{u}|=p^u|\mathcal{A}| 
\quad \text{and} \quad 
    \dim_{\mathbb{F}_p} \Span_{\mathbb{F}_p}(\mathcal{A}_{u}) =\dim_{\mathbb{F}_p} \Span_{\mathbb{F}_p}(\mathcal{A}) + u.  \end{equation}
\end{corollary}
\begin{proof}
We prove the assertion by induction on $u$. When $u=0$, we may take
$\mathcal{A}_{0}=\mathcal{A}$, and the assertion is immediate.

Now assume that $u\geq 1$ and that the assertion holds for $u-1$. That is, 
there exists an $\mathbb{F}_q$-good zero-set $\mathcal{A}_{u-1}\subseteq \mathbb{F}_{q^s}$
such that
\begin{equation}\notag
|\mathcal{A}_{u-1}|=p^{u-1}|\mathcal{A}|
\quad \hbox{and}\quad 
\dim_{\mathbb{F}_p}\Span_{\mathbb{F}_p}(\mathcal{A}_{u-1})
=
\dim_{\mathbb{F}_p}\Span_{\mathbb{F}_p}(\mathcal{A})+u-1.
\end{equation}
Since $u\leq es-r$ and
$
\dim_{\mathbb{F}_p}\Span_{\mathbb{F}_p}(\mathcal{A})\leq r,
$
we have
\begin{equation}\notag
\dim_{\mathbb{F}_p}\Span_{\mathbb{F}_p}(\mathcal{A}_{u-1})
\leq r+u-1
\leq es-1.
\end{equation}
Moreover,  since  $q=p^e$, we have
$
\dim_{\mathbb{F}_p}\mathbb{F}_{q^s}=es.
$
Thus, we have 
$
\Span_{\mathbb{F}_p}(\mathcal{A}_{u-1})\neq \mathbb{F}_{q^s}.
$
Then we can apply Theorem~\ref{theorem2} to obtain an
$\mathbb{F}_q$-good zero-set $\mathcal{A}_{u}\subseteq \mathbb{F}_{q^s}$ such that
\begin{equation}\notag
|\mathcal{A}_{u}|=p|\mathcal{A}_{u-1}|=p^u|\mathcal{A}|
\end{equation}
and
\begin{equation}\notag
\dim_{\mathbb{F}_p}\Span_{\mathbb{F}_p}(\mathcal{A}_{u})
=
\dim_{\mathbb{F}_p}\Span_{\mathbb{F}_p}(\mathcal{A}_{u-1})+1
=
\dim_{\mathbb{F}_p}\Span_{\mathbb{F}_p}(\mathcal{A})+u.
\end{equation}
Hence $\mathcal{A}_{u}$ satisfies \eqref{equation7}. This completes
the induction and hence the proof.
\end{proof}

\subsection{\texorpdfstring{$v$}{v}-th-power  construction} 
The proof of the following Lemma is given in the appendix.
\begin{lemma}\label{lemmas4}
Let $q=p^e$ for some prime $p$ and positive integer $e$,
let  $s$ and $h$ be positive integers such that $s\mid h$,  and let 
 $v$ be a positive divisor of $\frac{q^h-1}{q^s-1}$.  Then  for any 
$\alpha\in\mathbb{F}_{q^s}^{*}$, the polynomial  $x^v-\alpha$
has $v$ distinct nonzero roots in 
$\mathbb{F}_{q^h}$. Moreover, for  $\alpha, \alpha'\in \mathbb{F}_{q^s}^{*}$ with $\alpha\neq \alpha'$,  
we have 
\begin{equation}\label{dis}
    \{x\in \mathbb{F}_{q^h}:x^v-\alpha=0\}\cap \{x\in \mathbb{F}_{q^h}:x^v-\alpha'=0 \}=\varnothing.
\end{equation}
\end{lemma}

\begin{theorem}\label{theorem3}
Let $q=p^e$ for some prime $p$ and positive integer $e$,
let  $s$ and $h$ be positive integers such that $s\mid h$,  and let 
 $v$ be a positive divisor of $\frac{q^h-1}{q^s-1}$. 
Suppose that $\mathcal{A}\subseteq\mathbb{F}_{q^s}$ is  an $\mathbb{F}_q$-good zero-set.
Then 
\begin{equation}\notag
\mathcal{A}^{+}
=\{x\in\mathbb{F}_{q^h}:x^v\in\mathcal{A}\}
\end{equation}
is an $\mathbb{F}_q$-good zero-set satisfying
\begin{equation}\notag
|\mathcal{A}^{+}|=1+v(|\mathcal{A}|-1).
\end{equation}
\end{theorem}

\begin{proof}
Since $\mathcal{A}$ is an $\mathbb{F}_q$-good zero-set,
we have $0\in\mathcal{A}$ and hence $0\in\mathcal{A}^{+}$.
By Lemma~\ref{lemmas4}, for every
$\alpha\in\mathcal{A}\setminus\{0\}$, the equation $x^v=\alpha$
has exactly $v$ distinct nonzero solutions in $\mathbb{F}_{q^h}$.
Moreover, the solution sets corresponding to distinct values
of $\alpha$ are pairwise disjoint. 
Therefore, 
\begin{equation}\notag
|\mathcal{A}^{+}|=1+v(|\mathcal{A}|-1),
\end{equation}
and the vanishing polynomial of $\mathcal{A}^{+}$ is
\begin{equation}\notag
P_{\mathcal{A}^{+}}(x)
=x\prod_{\alpha\in\mathcal{A}\setminus\{0\}}(x^v-\alpha).
\end{equation}

It remains to show that $\mathcal{A}^{+}$ is an
$\mathbb{F}_q$-good zero-set.
Since $\mathcal{A}$ is an $\mathbb{F}_q$-good zero-set,
there exists $\lambda\in\mathbb{F}_{q^s}^{*}$ such that
\begin{equation}\notag
P_{\mathcal{A}}'(\alpha)\in\lambda\mathbb{F}_q^{*}
\qquad\text{for all }\alpha\in\mathcal{A}.
\end{equation}
A direct calculation yields
\begin{equation}\label{ke}
P_{\mathcal{A}^{+}}'(0) = \prod_{\alpha\in \mathcal{A}\setminus\{0\}}(-\alpha) = P_{\mathcal{A}}'(0) \in \lambda \mathbb{F}_q^*. 
\end{equation}
For any $\gamma\in\mathcal{A}^{+}\setminus\{0\}$, we have 
$\gamma^v\in\mathcal{A}\setminus\{0\}$, and 
\begin{equation}\notag
\begin{aligned}
P_{\mathcal{A}^{+}}'(\gamma)
&= v\gamma^v \prod_{\substack{\alpha\in \mathcal{A}\setminus\{0\} \\ \alpha\neq\gamma^v}} (\gamma^v-\alpha) \\
&= vP_{\mathcal{A}}'(\gamma^v)\in \lambda v\mathbb{F}_q^*.
\end{aligned} 
\end{equation}
Moreover, since $v$ is a divisor of $\frac{q^h-1}{q^s-1}$ and $q=p^e$,  we have $p\nmid v$ and hence $v\mathbb{F}_q^{*}=\mathbb{F}_q^{*}$.
Thus
\begin{equation}\notag
P_{\mathcal{A}^{+}}'(\gamma)\in\lambda\mathbb{F}_q^{*}
\qquad\text{for all }\gamma\in\mathcal{A}^{+}\setminus\{0\}.
\end{equation}
Together with $0\in\mathcal{A}^{+}$ and \eqref{ke}, this proves that
$\mathcal{A}^{+}$ is an $\mathbb{F}_q$-good zero-set.
\end{proof}

\subsection{Shifted inverse construction}
\begin{theorem}\label{theorem4}
Let $q=p^e$ for some prime $p$ and positive integer $e$, let
$\mathbb{F}$ be a finite extension of $\mathbb{F}_q$, and let
$\mathcal{A}\subseteq\mathbb{F}$ be a strict
$\mathbb{F}_q$-good zero-set. Suppose that there exists
$\beta\in\mathbb{F}\setminus\mathcal{A}$ such that
\begin{equation}\label{equation12}
(\alpha-\beta)^{|\mathcal{A}|-1}\in\mathbb{F}_q^*
\qquad\text{for all }\alpha\in\mathcal{A}.
\end{equation}
Then
\begin{equation}\notag
\mathcal{A}^{+}
=
\{0\}\cup
\left\{
\frac{1}{\alpha-\beta}:\alpha\in\mathcal{A}
\right\}
\end{equation}
is an $\mathbb{F}_q$-good zero-set satisfying
$
|\mathcal{A}^{+}|=|\mathcal{A}|+1.
$ 
\end{theorem}

\begin{proof}
For each $\alpha\in\mathcal{A}$, set
$
y_\alpha=\frac{1}{\alpha-\beta}.
$ 
Since $\beta\notin\mathcal{A}$, each  $y_\alpha$ is well defined and
nonzero. Moreover, for $\alpha,\alpha'\in\mathcal{A}$,
the equality $y_\alpha=y_{\alpha'}$ implies $\alpha=\alpha'$.
Thus the elements $y_\alpha$ are pairwise distinct. Therefore,
\begin{equation}\notag
|\mathcal{A}^{+}|=|\mathcal{A}|+1,
\end{equation}
and the vanishing polynomial of $\mathcal{A}^{+}$ is
\begin{equation}\notag
P_{\mathcal{A}^{+}}(x)
=
x\prod_{\alpha\in\mathcal{A}}(x-y_\alpha).
\end{equation}

Since $\beta\in \mathbb{F}\setminus \mathcal{A}$,
a direct calculation gives
\begin{equation}\notag
P_{\mathcal{A}^{+}}'(0)
=
\prod_{\theta\in\mathcal{A}}(-y_\theta)
=
\prod_{\theta\in\mathcal{A}}\frac{1}{\beta-\theta}
\neq0.
\end{equation}
For any  $\alpha\in\mathcal{A}$, we  have
\begin{equation}\notag
P_{\mathcal{A}^{+}}'(y_\alpha)
=
y_\alpha
\prod_{\substack{\theta\in\mathcal{A}\\\theta\neq\alpha}}
(y_\alpha-y_\theta).
\end{equation}
Consequently,
\begin{equation}\notag
\begin{aligned}
\frac{P_{\mathcal{A}^{+}}'(y_\alpha)}
     {P_{\mathcal{A}^{+}}'(0)}
&=
-\prod_{\substack{\theta\in\mathcal{A}\\\theta\neq\alpha}}
\frac{y_\alpha-y_\theta}{-y_\theta} \\
&=
-\prod_{\substack{\theta\in\mathcal{A}\\\theta\neq\alpha}}
\frac{\alpha-\theta}{\alpha-\beta} \\
&=
-\frac{P_{\mathcal{A}}'(\alpha)}
        {(\alpha-\beta)^{|\mathcal{A}|-1}}.
\end{aligned}
\end{equation}
Since $\mathcal{A}$ is a strict $\mathbb{F}_q$-good zero-set, we have
\begin{equation}\notag
P_{\mathcal{A}}'(\alpha)\in\mathbb{F}_q^*
\qquad\text{for all }\alpha\in\mathcal{A}.
\end{equation}
Together with \eqref{equation12}, this yields
\begin{equation}\notag
P_{\mathcal{A}^{+}}'(y_\alpha)
\in
P_{\mathcal{A}^{+}}'(0)\mathbb{F}_q^*
\qquad\text{for all }\alpha\in\mathcal{A}.
\end{equation}
It is clear that $P_{\mathcal{A}^{+}}'(0)
\in
P_{\mathcal{A}^{+}}'(0)\mathbb{F}_q^*$. Hence
\begin{equation}\notag
P_{\mathcal{A}^{+}}'(\gamma)
\in
P_{\mathcal{A}^{+}}'(0)\mathbb{F}_q^*
\qquad\text{for all }\gamma\in\mathcal{A}^{+}.
\end{equation}
Therefore, $\mathcal{A}^{+}$ is an
$\mathbb{F}_q$-good zero-set.
\end{proof}

\subsection{Special polynomial construction}
For integers $a$ and $b$, we write $a\mid b$ if $a$ divides $b$,
and $a\nmid b$ otherwise. We can obtain $\mathbb{F}_q$-good zero-sets from some special polynomials,
as the following result shows.
\begin{theorem}\label{theorem5}
Let $q=p^e$ for some prime $p$ and positive integer $e$, and let
$\mathbb{F}$ be a finite extension of $\mathbb{F}_q$.
Let $\beta\in\mathbb{F}^*$, let $r$ be a positive integer such that
$p\nmid r$, and let $H(x)\in\mathbb{F}[x]$ be a nonzero monic polynomial.
Suppose that
\begin{equation}\notag
U(x)=\beta+x^rH(x)^p
\end{equation}
 splits completely over $\mathbb{F}$.
Then
\begin{equation}\notag
\mathcal{A}
=
\{0\}\cup\{\alpha\in\mathbb{F}:U(\alpha)=0\}
\end{equation}
is an $\mathbb{F}_q$-good zero-set satisfying
$
|\mathcal{A}|=\deg(U)+1.
$ 
\end{theorem}
\begin{proof}
Differentiating $U(x)$ gives
\begin{equation}\notag
U'(x)=rx^{r-1}H(x)^p.
\end{equation}
For any root $\alpha\in\mathbb F$ of $U(x)$, we have
\begin{equation}\notag
\alpha^rH(\alpha)^p=-\beta.
\end{equation}
Since $\beta\neq0$, it follows that $\alpha\neq0$. Moreover,
$p\nmid r$ implies that
\begin{equation}\notag
U'(\alpha)
=r\alpha^{r-1}H(\alpha)^p
=-\frac{r\beta}{\alpha}
\neq0.
\end{equation}
Thus all roots of $U(x)$ are nonzero and simple.
Since $U(x)$ splits completely over $\mathbb F$, we obtain
\begin{equation}\notag
|\mathcal A|=\deg(U)+1.
\end{equation}
Furthermore, $U(x)$ is monic, so the vanishing polynomial of
$\mathcal A$ is
\begin{equation}\notag
P_{\mathcal A}(x)=xU(x).
\end{equation}
It follows that
\begin{equation}\label{ke1}
P_{\mathcal A}'(0)=U(0)=\beta.
\end{equation}
For every $\alpha\in\mathcal A\setminus\{0\}$, we have
$U(\alpha)=0$, and hence
\begin{equation}\label{ke2}
P_{\mathcal A}'(\alpha)
=U(\alpha)+\alpha U'(\alpha)
=-r\beta.
\end{equation}
Since $p\nmid r$, we have
$-r\in\mathbb F_p^*\subseteq\mathbb F_q^*$.
Combining \eqref{ke1} and \eqref{ke2}, we obtain
\begin{equation}\notag
P_{\mathcal A}'(\alpha)\in\beta\mathbb F_q^*
\qquad\text{for all }\alpha\in\mathcal A.
\end{equation}
Together with $0\in\mathcal A$, this shows that $\mathcal A$
is an $\mathbb F_q$-good zero-set.
This completes the proof.
\end{proof}

\section{Families of primitive narrow-sense BCH codes attaining their designed distances} 
\subsection{Designed distances of the form
\texorpdfstring{$\delta=\bigl(1+v(tp^u-1)\bigr)p^b-1$}
{ }}
We first establish the existence of $\mathbb{F}_q$-good zero-sets
of cardinality $t$ for each integer $1\le t\le q$.
These sets serve as starting points for constructing larger good
zero-sets of the required cardinalities in this subsection.
The proof of the following lemma is given in the appendix.
\begin{lemma}\label{lemma5}
Let $q=p^e$ for some prime $p$ and positive integer $e$, and 
let $t$ be an integer with 
$1\le t\le q$. Then there exists a strict $\mathbb{F}_q$-good zero-set $\mathcal{A}\subseteq\F_q$ such that
\begin{equation}\notag
        \qquad |\mathcal{A}|=t
        \quad \hbox{and}\quad  \dim_{\F_p}\Span_{\F_p}(\mathcal{A})\leq \lceil\log_p t\rceil.
\end{equation}
\end{lemma}

\begin{theorem}\label{theorem6}
Let $q=p^e$ for some prime $p$ and positive integer $e$, and  let  $s$, $h$, and $m$ be positive
integers such that $s\mid h$ and $h\mid m$. 
Suppose that $\delta\geq 2$ is an integer of the form
\begin{equation}\notag
\delta=\bigl(1+v(tp^u-1)\bigr)p^b-1,
\end{equation}
where $v$ is a positive divisor of $\frac{q^h-1}{q^s-1}$, 
and $t$, $u$ and $b$ are integers satisfying
$1\leq t\leq q$, $0\leq u\leq es-\lceil\log_p t\rceil$ and
$0\leq b\leq e(m-h)$.
Then
\begin{equation}\notag
d(\mathcal{C}_{(q,m,\delta)})=\delta.
\end{equation}
\end{theorem}
\begin{proof}
    By Lemma~\ref{lemma5}, there exists an $\mathbb{F}_q$-good zero-set $\mathcal{A}_{0}\subseteq \mathbb{F}_q$ such that \begin{equation}\notag
   |\mathcal{A}_{0}|=t\quad \hbox{and}\quad \dim_{\F_p}\Span_{\F_p}(\mathcal{A}_{0})\leq \lceil\log_p t\rceil. 
    \end{equation}
Since $0\leq u\leq es- \lceil\log_p t\rceil$,  Corollary \ref{corollary2} yields
an $\mathbb{F}_q$-good zero-set $\mathcal{A}_{u}\subseteq \mathbb{F}_{q^s}$ such that \begin{equation}\notag
|\mathcal{A}_{u}|=p^u|\mathcal{A}_{0}|=tp^u.
\end{equation}
Since $s\mid h$ and  $v\mid \frac{q^h-1}{q^s-1}$, we can apply Theorem \ref{theorem3} to conclude that 
\begin{equation}\notag
\mathcal{A}^{+}
=\{x\in\mathbb{F}_{q^h}:x^v\in\mathcal{A}_u\}
\end{equation} is an $\mathbb{F}_q$-good zero-set satisfying \begin{equation}\notag
|\mathcal{A}^{+}|=1+v(|\mathcal{A}_{u}|-1)=1+v(tp^u-1).
\end{equation}
Note that $\mathcal{A}^{+}\subseteq \mathbb{F}_{q^h}$ and $h\mid m$.   Thus, we obtain \begin{equation}\notag
\mathcal{A}^{+}\subseteq\mathbb{F}_{q^m}\quad \hbox{and}\quad \dim_{\mathbb{F}_p} \Span_{\mathbb{F}_p}(\mathcal{A}^{+})\leq \dim_{\mathbb{F}_p}(\mathbb{F}_{q^h}) = eh.
\end{equation}
Since $0\leq b\leq em-eh$, 
we can apply Corollary~\ref{corollary2} again to obtain an $\mathbb{F}_q$-good zero-set $\mathcal{A}^{++}\subseteq \mathbb{F}_{q^m}$ such that
\begin{equation}\notag
|\mathcal{A}^{++}|=|\mathcal{A}^{+}|p^b=\bigl(1+v(tp^u-1)\bigr)p^b=\delta+1.
\end{equation}
Then applying Theorem \ref{theorem1}, we have $d(\mathcal{C}_{(q,m,\delta)})=\delta$.  This completes the proof.
\end{proof}
It is worth noting that Theorem~\ref{theorem6} covers a broader range of designed distances
while recovering several known results as special cases. The following
corollary records some of these cases. In particular, part~(ii) is the
$q$-power specialization of part~(i). Part~(ii)
recovers the  minimum-distance  results in
\cite[Theorem~5, p.260]{MacWilliamsSloane} and 
\cite[Theorem~10]{D2015}. 
Parts~(iii) and~(iv) recover the primitive case of 
\cite[Lemma~19]{LiuDingLi2017} and the case $ 2a \mid  m$ of 
\cite[Theorem~13]{DingDuZhou2015}, respectively.
\begin{corollary}\label{corollary:known-cases}
Let $q=p^e$ for some prime $p$ and positive integer $e$,
and let $m$ be a positive integer.
Suppose that $\delta\geq2$ has one of the following forms:
\begin{enumerate}[label=(\roman*)]
\item 
$\delta=tp^u-1$,
where $t$ and $u$ are integers satisfying 
$1\leq t\leq q$ and
$0\leq u\leq em-\lceil\log_p t\rceil$.

\item 
$\delta=tq^a-1$,
where $t$ and $a$ are integers  satisfying 
$1\leq t\leq q$ and $0\leq a\leq m-1$.
\item
$\delta=kv$,
where $k$ and $v$ are positive integers satisfying
$v\mid(q^m-1)/(q-1)$ and $1\leq k\leq q-1$.

\item 
$\delta=q^a+1$,
where $a$ is an integer satisfying  $2a\mid m$.
\end{enumerate}
Then
\begin{equation}\notag
d\bigl(\mathcal{C}_{(q,m,\delta)}\bigr)=\delta.
\end{equation}
\end{corollary}

\begin{proof}
Part~(i) follows from Theorem~\ref{theorem6} by taking
$s=h=m$, $v=1$, and $b=0$. 
 Since
$1\leq t\leq q$, we have
$\lceil\log_p t\rceil\leq  \lceil\log_p q\rceil =e$. Consequently, for
$0\leq a\leq m-1$,
\begin{equation}\notag
0\leq ea
\leq e(m-1)
\leq em-\left\lceil\log_p t\right\rceil.
\end{equation}
Thus, we can take $u=ae$ in part~(i) to obtain part~(ii). 

Applying  Theorem~\ref{theorem6} with
$s=1$, $h=m$, $t=k+1$, and $u=b=0$, we obtain part (iii). 
Finally, suppose that the conditions in part (iv) hold.
Notice that
\begin{equation}\notag
\frac{q^m-1}{q-1}
=
\frac{q^m-1}{q^{2a}-1}
\frac{q^a-1}{q-1}
(q^a+1)
\end{equation}
and $\frac{q^a-1}{q-1}$ is an integer. 
Since $2a\mid m$, 
$\frac{q^m-1}{q^{2a}-1}$ is also an integer. 
Hence $q^a+1$ divides $ \frac{q^m-1}{q-1}$.
Taking $k=1$ and $v=q^a+1$ in part~(iii) gives
$\delta=q^a+1$. This completes the proof.
\end{proof}

A $q$-ary linear code is called \emph{optimal} if no $q$-ary linear
code with the same length and dimension has a larger minimum distance.
Table~\ref{tab:examples-theorem-6} presents selected examples of
the BCH codes in Theorem~\ref{theorem6}.
The dimensions were computed and the minimum distances verified
using Magma.
Comparison with the Database shows that each listed code is either optimal or attains the best-known minimum distance.

\begingroup\scriptsize\setlength{\tabcolsep}{4pt}\renewcommand{\arraystretch}{0.85}
\begin{longtable}[c]{@{}rrrrrrcc@{}}
\caption{Examples of BCH codes in Theorem 6}\label{tab:examples-theorem-6}\\
\toprule
\multicolumn{1}{@{}c}{$q$} & \multicolumn{1}{c}{$m$} & \multicolumn{1}{c}{$\delta$} & \multicolumn{1}{c}{$n$} & \multicolumn{1}{c}{$\dim$} & \multicolumn{1}{c}{$d$} & \multicolumn{1}{c}{Parameter choices} & \multicolumn{1}{c@{}}{optimality}\\
\midrule
\endfirsthead
\multicolumn{8}{c}{\tablename\ \thetable\ -- continued from previous page}\\
\toprule
\multicolumn{1}{@{}c}{$q$} & \multicolumn{1}{c}{$m$} & \multicolumn{1}{c}{$\delta$} & \multicolumn{1}{c}{$n$} & \multicolumn{1}{c}{$\dim$} & \multicolumn{1}{c}{$d$} & \multicolumn{1}{c}{Parameter choices} & \multicolumn{1}{c@{}}{optimality}\\
\midrule
\endhead
\midrule
\multicolumn{8}{r}{Continued on next page}\\
\endfoot
\bottomrule
\endlastfoot
2 & 2 & 3 & 3 & 1 & 3 & $ (s,h,t,u,b,v)=(1,1,1,1,1,1) $ & optimal\\
2 & 3 & 3 & 7 & 4 & 3 & $ (s,h,t,u,b,v)=(1,1,1,0,2,1) $ & optimal\\
2 & 3 & 7 & 7 & 1 & 7 & $ (s,h,t,u,b,v)=(1,1,1,1,2,1) $ & optimal\\
2 & 4 & 3 & 15 & 11 & 3 & $ (s,h,t,u,b,v)=(1,1,1,0,2,1) $ & optimal\\
2 & 4 & 5 & 15 & 7 & 5 & $ (s,h,t,u,b,v)=(1,4,1,1,0,5) $ & optimal\\
2 & 4 & 7 & 15 & 5 & 7 & $ (s,h,t,u,b,v)=(1,1,1,0,3,1) $ & optimal\\
2 & 4 & 15 & 15 & 1 & 15 & $ (s,h,t,u,b,v)=(1,1,1,1,3,1) $ & optimal\\
2 & 5 & 3 & 31 & 26 & 3 & $ (s,h,t,u,b,v)=(1,1,1,0,2,1) $ & optimal\\
2 & 5 & 15 & 31 & 6 & 15 & $ (s,h,t,u,b,v)=(1,1,1,0,4,1) $ & optimal\\
2 & 5 & 31 & 31 & 1 & 31 & $ (s,h,t,u,b,v)=(1,1,1,1,4,1) $ & optimal\\
2 & 6 & 3 & 63 & 57 & 3 & $ (s,h,t,u,b,v)=(1,1,1,0,2,1) $ & optimal\\
2 & 6 & 9 & 63 & 39 & 9 & $ (s,h,t,u,b,v)=(1,6,1,1,0,9) $ & best known\\
2 & 6 & 21 & 63 & 18 & 21 & $ (s,h,t,u,b,v)=(1,6,1,1,0,21) $ & best known\\
2 & 6 & 27 & 63 & 10 & 27 & $ (s,h,t,u,b,v)=(3,6,1,2,0,9) $ & best known\\
2 & 6 & 31 & 63 & 7 & 31 & $ (s,h,t,u,b,v)=(1,1,1,0,5,1) $ & optimal\\
2 & 6 & 63 & 63 & 1 & 63 & $ (s,h,t,u,b,v)=(1,1,1,1,5,1) $ & optimal\\
2 & 7 & 3 & 127 & 120 & 3 & $ (s,h,t,u,b,v)=(1,1,1,0,2,1) $ & optimal\\
2 & 7 & 7 & 127 & 106 & 7 & $ (s,h,t,u,b,v)=(1,1,1,0,3,1) $ & best known\\
2 & 7 & 15 & 127 & 78 & 15 & $ (s,h,t,u,b,v)=(1,1,1,0,4,1) $ & best known\\
2 & 7 & 63 & 127 & 8 & 63 & $ (s,h,t,u,b,v)=(1,1,1,0,6,1) $ & optimal\\
2 & 7 & 127 & 127 & 1 & 127 & $ (s,h,t,u,b,v)=(1,1,1,1,6,1) $ & optimal\\
2 & 8 & 3 & 255 & 247 & 3 & $ (s,h,t,u,b,v)=(1,1,1,0,2,1) $ & optimal\\
2 & 8 & 5 & 255 & 239 & 5 & $ (s,h,t,u,b,v)=(1,4,1,1,0,5) $ & optimal\\
2 & 8 & 7 & 255 & 231 & 7 & $ (s,h,t,u,b,v)=(1,1,1,0,3,1) $ & best known\\
2 & 8 & 11 & 255 & 215 & 11 & $ (s,h,t,u,b,v)=(1,4,1,1,1,5) $ & best known\\
2 & 8 & 15 & 255 & 199 & 15 & $ (s,h,t,u,b,v)=(1,1,1,0,4,1) $ & best known\\
2 & 8 & 17 & 255 & 191 & 17 & $ (s,h,t,u,b,v)=(1,8,1,1,0,17) $ & best known\\
2 & 8 & 23 & 255 & 171 & 23 & $ (s,h,t,u,b,v)=(1,4,1,1,2,5) $ & best known\\
2 & 8 & 31 & 255 & 139 & 31 & $ (s,h,t,u,b,v)=(1,1,1,0,5,1) $ & best known\\
2 & 8 & 47 & 255 & 99 & 47 & $ (s,h,t,u,b,v)=(1,4,1,1,3,5) $ & best known\\
2 & 8 & 51 & 255 & 91 & 51 & $ (s,h,t,u,b,v)=(1,8,1,1,0,51) $ & best known\\
2 & 8 & 85 & 255 & 47 & 85 & $ (s,h,t,u,b,v)=(1,8,1,1,0,85) $ & best known\\
2 & 8 & 119 & 255 & 13 & 119 & $ (s,h,t,u,b,v)=(4,8,1,3,0,17) $ & best known\\
2 & 8 & 127 & 255 & 9 & 127 & $ (s,h,t,u,b,v)=(1,1,1,0,7,1) $ & optimal\\
2 & 8 & 255 & 255 & 1 & 255 & $ (s,h,t,u,b,v)=(1,1,1,1,7,1) $ & optimal\\
3 & 1 & 2 & 2 & 1 & 2 & $ (s,h,t,u,b,v)=(1,1,1,1,0,1) $ & optimal\\
3 & 2 & 2 & 8 & 6 & 2 & $ (s,h,t,u,b,v)=(1,1,1,0,1,1) $ & optimal\\
3 & 2 & 4 & 8 & 4 & 4 & $ (s,h,t,u,b,v)=(1,2,1,1,0,2) $ & optimal\\
3 & 2 & 5 & 8 & 3 & 5 & $ (s,h,t,u,b,v)=(1,1,2,0,1,1) $ & optimal\\
3 & 2 & 8 & 8 & 1 & 8 & $ (s,h,t,u,b,v)=(1,1,1,1,1,1) $ & optimal\\
3 & 3 & 2 & 26 & 23 & 2 & $ (s,h,t,u,b,v)=(1,1,1,0,1,1) $ & optimal\\
3 & 3 & 13 & 26 & 8 & 13 & $ (s,h,t,u,b,v)=(1,3,2,0,0,13) $ & optimal\\
3 & 3 & 17 & 26 & 4 & 17 & $ (s,h,t,u,b,v)=(1,1,2,0,2,1) $ & optimal\\
3 & 3 & 26 & 26 & 1 & 26 & $ (s,h,t,u,b,v)=(1,1,1,1,2,1) $ & optimal\\
3 & 4 & 2 & 80 & 76 & 2 & $ (s,h,t,u,b,v)=(1,1,1,0,1,1) $ & optimal\\
3 & 4 & 4 & 80 & 72 & 4 & $ (s,h,t,u,b,v)=(1,2,1,1,0,2) $ & optimal\\
3 & 4 & 8 & 80 & 60 & 8 & $ (s,h,t,u,b,v)=(1,1,1,0,2,1) $ & best known\\
3 & 4 & 10 & 80 & 56 & 10 & $ (s,h,t,u,b,v)=(1,4,1,1,0,5) $ & best known\\
3 & 4 & 14 & 80 & 46 & 14 & $ (s,h,t,u,b,v)=(1,2,1,1,1,2) $ & best known\\
3 & 4 & 40 & 80 & 16 & 40 & $ (s,h,t,u,b,v)=(1,4,1,1,0,20) $ & best known\\
3 & 4 & 44 & 80 & 11 & 44 & $ (s,h,t,u,b,v)=(1,2,1,1,2,2) $ & best known\\
3 & 4 & 50 & 80 & 7 & 50 & $ (s,h,t,u,b,v)=(2,4,2,1,0,10) $ & optimal\\
3 & 4 & 53 & 80 & 5 & 53 & $ (s,h,t,u,b,v)=(1,1,2,0,3,1) $ & optimal\\
3 & 4 & 80 & 80 & 1 & 80 & $ (s,h,t,u,b,v)=(1,1,1,1,3,1) $ & optimal\\
3 & 5 & 2 & 242 & 237 & 2 & $ (s,h,t,u,b,v)=(1,1,1,0,1,1) $ & optimal\\
3 & 5 & 5 & 242 & 227 & 5 & $ (s,h,t,u,b,v)=(1,1,2,0,1,1) $ & best known\\
3 & 5 & 8 & 242 & 217 & 8 & $ (s,h,t,u,b,v)=(1,1,1,0,2,1) $ & best known\\
3 & 5 & 11 & 242 & 207 & 11 & $ (s,h,t,u,b,v)=(1,5,2,0,0,11) $ & best known\\
3 & 5 & 26 & 242 & 157 & 26 & $ (s,h,t,u,b,v)=(1,1,1,0,3,1) $ & best known\\
3 & 5 & 121 & 242 & 32 & 121 & $ (s,h,t,u,b,v)=(1,5,2,0,0,121) $ & best known\\
3 & 5 & 161 & 242 & 6 & 161 & $ (s,h,t,u,b,v)=(1,1,2,0,4,1) $ & optimal\\
3 & 5 & 242 & 242 & 1 & 242 & $ (s,h,t,u,b,v)=(1,1,1,1,4,1) $ & optimal\\
4 & 1 & 2 & 3 & 2 & 2 & $ (s,h,t,u,b,v)=(1,1,3,0,0,1) $ & optimal\\
4 & 1 & 3 & 3 & 1 & 3 & $ (s,h,t,u,b,v)=(1,1,1,2,0,1) $ & optimal\\
4 & 2 & 2 & 15 & 13 & 2 & $ (s,h,t,u,b,v)=(1,1,3,0,0,1) $ & optimal\\
4 & 2 & 5 & 15 & 9 & 5 & $ (s,h,t,u,b,v)=(1,1,3,0,1,1) $ & optimal\\
4 & 2 & 10 & 15 & 4 & 10 & $ (s,h,t,u,b,v)=(1,2,3,0,0,5) $ & optimal\\
4 & 2 & 11 & 15 & 3 & 11 & $ (s,h,t,u,b,v)=(1,1,3,0,2,1) $ & optimal\\
4 & 2 & 15 & 15 & 1 & 15 & $ (s,h,t,u,b,v)=(1,1,1,2,2,1) $ & optimal\\
4 & 3 & 2 & 63 & 60 & 2 & $ (s,h,t,u,b,v)=(1,1,3,0,0,1) $ & optimal\\
4 & 3 & 5 & 63 & 54 & 5 & $ (s,h,t,u,b,v)=(1,1,3,0,1,1) $ & best known\\
4 & 3 & 6 & 63 & 51 & 6 & $ (s,h,t,u,b,v)=(1,3,3,0,0,3) $ & best known\\
4 & 3 & 9 & 63 & 45 & 9 & $ (s,h,t,u,b,v)=(1,3,1,2,0,3) $ & best known\\
4 & 3 & 21 & 63 & 27 & 21 & $ (s,h,t,u,b,v)=(1,3,1,1,0,21) $ & best known\\
4 & 3 & 42 & 63 & 8 & 42 & $ (s,h,t,u,b,v)=(1,3,3,0,0,21) $ & optimal\\
4 & 3 & 47 & 63 & 4 & 47 & $ (s,h,t,u,b,v)=(1,1,3,0,4,1) $ & optimal\\
4 & 3 & 63 & 63 & 1 & 63 & $ (s,h,t,u,b,v)=(1,1,1,2,4,1) $ & optimal\\
4 & 4 & 2 & 255 & 251 & 2 & $ (s,h,t,u,b,v)=(1,1,3,0,0,1) $ & optimal\\
4 & 4 & 5 & 255 & 243 & 5 & $ (s,h,t,u,b,v)=(1,1,3,0,1,1) $ & best known\\
4 & 4 & 15 & 255 & 211 & 15 & $ (s,h,t,u,b,v)=(1,1,1,0,4,1) $ & best known\\
4 & 4 & 170 & 255 & 16 & 170 & $ (s,h,t,u,b,v)=(1,4,3,0,0,85) $ & best known\\
4 & 4 & 175 & 255 & 11 & 175 & $ (s,h,t,u,b,v)=(1,2,3,0,4,5) $ & best known\\
4 & 4 & 187 & 255 & 7 & 187 & $ (s,h,t,u,b,v)=(2,4,3,2,0,17) $ & best known\\
4 & 4 & 191 & 255 & 5 & 191 & $ (s,h,t,u,b,v)=(1,1,3,0,6,1) $ & optimal\\
4 & 4 & 255 & 255 & 1 & 255 & $ (s,h,t,u,b,v)=(1,1,1,2,6,1) $ & optimal\\
\end{longtable}\endgroup

\subsection{Designed distances of the form
\texorpdfstring{$\delta=
\bigl(1+v\bigl((q^t+2)p^u-1\bigr)\bigr)p^b-1$}
{ }}
 The following result of Chen et al.~\cite{Chen2026} determines
the minimum distance of $\mathcal{C}_{(q,m,q^t+1)}$ and hence
yields $\mathbb{F}_q$-good zero-sets of cardinality $q^t+2$
under suitable conditions.
\begin{lemma}\label{lemma6}
\textnormal{\cite[Theorem 8]{Chen2026}}
Let $q$ be a prime power, and 
let $m$ and $t$ be two positive integers such that
$t\mid m$ and $t<m$. Then
\begin{equation}\notag
d\left(\mathcal{C}_{(q,m,q^t+1)}\right)=q^t+1.
\end{equation}
\end{lemma}

\begin{theorem}
\label{theorem7}
Let $q=p^e$ for some prime $p$ and positive integer $e$, and 
let  $l,s,h,m$ be positive integers such that
$
l\mid s,  s\mid h$  and $h\mid m
$.  Suppose that  $\delta\geq 2$ is  an integer of the form 
\begin{equation}\notag
\delta
=
\left[
1+v\left((q^t+2)p^u-1\right)
\right]p^b-1,
\end{equation}
where $v$ is a positive divisor of $\frac{q^h-1}{q^s-1}$, and   $t,u, b$ are nonnegative  integers satisfying $t\mid l$, $1\leq t<l$, 
$
0\leq u\leq e(s-l) $ and $
0\leq b\leq e(m-h).
$
 Then 
\begin{equation}\notag
d\left(\mathcal{C}_{(q,m,\delta)}\right)=\delta.
\end{equation}
\end{theorem}
\begin{proof}
Since $t\mid l$ and $t<l$, Lemma~\ref{lemma6} gives
$d\left(\mathcal{C}_{(q,l,q^t+1)}\right)=q^t+1$. 
Then by Theorem~\ref{theorem1}, there exists an
$\mathbb F_q$-good zero-set
$
\mathcal{A}_{0}\subseteq\mathbb F_{q^{l}}
$
such that
\begin{equation}\notag
|\mathcal{A}_{0}|=q^t+2.
\end{equation}
Moreover, 
\begin{equation}\notag
\dim_{\mathbb F_p}
\Span_{\mathbb F_p}(\mathcal{A}_{0})
\leq
\dim_{\mathbb F_p}\mathbb F_{q^{l}}
=
el.
\end{equation}
Since 
$
l\mid s 
$, we have $\mathbb{F}_{q^l}\subseteq \mathbb{F}_{q^s}$ and hence $\mathcal{A}_0\subseteq \mathbb{F}_{q^s}$.

 Since  
$
0\leq u\leq es-el,
$ we can apply  
Corollary \ref{corollary2} with $r=el$ to obtain   an $\mathbb F_q$-good zero-set
$
\mathcal{A}_{u}\subseteq\mathbb F_{q^s}
$
satisfying 
\begin{equation}\notag
|\mathcal{A}_{u}|
=
p^u|\mathcal{A}_{0}|
=
(q^t+2)p^u.
\end{equation}
Furthermore, 
since
$s\mid h$ and
$
v\mid\frac{q^h-1}{q^s-1},
$  applying  
Theorem \ref{theorem3} 
yields   an $\mathbb F_q$-good zero-set
$
\mathcal{A}^{+}\subseteq\mathbb F_{q^h}
$
satisfying 
\begin{equation}\notag
\begin{aligned}
|\mathcal{A}^{+}|
=
1+v\left(|\mathcal{A}_{u}|-1\right)=
1+v\left((q^t+2)p^u-1\right)
\end{aligned}
\end{equation}
and 
\begin{equation}\notag
\dim_{\mathbb{F}_p} \Span_{\mathbb{F}_p}(\mathcal{A}^{+})\leq \dim_{\mathbb{F}_p}(\mathbb{F}_{q^h}) = eh .
\end{equation}
Since 
$h\mid m$, we have $\mathbb{F}_{q^h}\subseteq \mathbb{F}_{q^m}$ and hence $\mathcal{A}^{+}\subseteq \mathbb{F}_{q^m}$. 
Since $0\leq b\leq em-eh$, we can apply Corollary~\ref{corollary2} again to
the field $\mathbb F_{q^m}$ to get an 
$\mathbb F_q$-good zero-set
$
\mathcal{A}^{++}\subseteq\mathbb F_{q^m}
$ satisfying
\begin{equation}\notag
\begin{aligned}
|\mathcal{A}^{++}|=
p^b|\mathcal{A}^{+}|=
\left[
1+v\left((q^t+2)p^u-1\right)
\right]p^b=
\delta+1.
\end{aligned}
\end{equation}
By Theorem \ref{theorem1}, we obtain $
d\left(\mathcal{C}_{(q,m,\delta)}\right)=\delta.
$ 
This completes the proof.
\end{proof}

The following corollary records several simple special cases of
Theorem~\ref{theorem7}.
\begin{corollary}\label{corollary:theorem7-cases}
Let $q=p^e$ for some prime $p$ and positive integer $e$,
and let $l$  and $m$ be  positive integers with $l\mid m$, and let $t$ be a positive integer with $t\mid l$ and $t<l$. 
Suppose that $\delta\geq2$ has one of the following forms:
\begin{enumerate}
\item[(i)]
$\delta=(q^t+2)p^u-1$ 
for some integer  $u$ with  $0\leq u\leq e(m-l)$.

\item[(ii)]
$\delta=(q^t+2)q^a-1 $
for some integer $a$ with  $0\leq a\leq m-l$.

\item[(iii)]
$\delta=v(q^t+1)$ 
for some positive divisor $v$ of 
$\frac{q^m-1}{q^l-1}.
$
\end{enumerate}
Then
\begin{equation}\notag
d\bigl(\mathcal{C}_{(q,m,\delta)}\bigr)=\delta.
\end{equation}
\end{corollary}

\begin{proof}
Taking $s=h=m$, $v=1$, and $b=0$ in Theorem~\ref{theorem7} yields part (i). 
Note  that $q=p^e$. 
Part (ii) follows from part (i) by setting $u=ea$.  
Finally, applying Theorem~\ref{theorem7} with $s=l$, $h=m$, and $u=b=0$ gives part (iii).
\end{proof}

\begingroup\scriptsize\setlength{\tabcolsep}{4pt}\renewcommand{\arraystretch}{0.85}
\begin{longtable}[c]{@{}rrrrrrcc@{}}
\caption{Examples of BCH codes in Theorem 7}\label{tab:examples-theorem-7}\\
\toprule
\multicolumn{1}{@{}c}{$q$} & \multicolumn{1}{c}{$m$} & \multicolumn{1}{c}{$\delta$} & \multicolumn{1}{c}{$n$} & \multicolumn{1}{c}{$\dim$} & \multicolumn{1}{c}{$d$} & \multicolumn{1}{c}{Parameter choices} & \multicolumn{1}{c@{}}{optimality}\\
\midrule
\endfirsthead
\multicolumn{8}{c}{\tablename\ \thetable\ -- continued from previous page}\\
\toprule
\multicolumn{1}{@{}c}{$q$} & \multicolumn{1}{c}{$m$} & \multicolumn{1}{c}{$\delta$} & \multicolumn{1}{c}{$n$} & \multicolumn{1}{c}{$\dim$} & \multicolumn{1}{c}{$d$} & \multicolumn{1}{c}{Parameter choices} & \multicolumn{1}{c@{}}{optimality}\\
\midrule
\endhead
\midrule
\multicolumn{8}{r}{Continued on next page}\\
\endfoot
\bottomrule
\endlastfoot
2 & 2 & 3 & 3 & 1 & 3 & $ (l,s,h,t,u,b,v)=(2,2,2,1,0,0,1) $ & optimal\\
2 & 3 & 3 & 7 & 4 & 3 & $ (l,s,h,t,u,b,v)=(3,3,3,1,0,0,1) $ & optimal\\
2 & 4 & 3 & 15 & 11 & 3 & $ (l,s,h,t,u,b,v)=(2,2,2,1,0,0,1) $ & optimal\\
2 & 4 & 5 & 15 & 7 & 5 & $ (l,s,h,t,u,b,v)=(4,4,4,2,0,0,1) $ & optimal\\
2 & 4 & 7 & 15 & 5 & 7 & $ (l,s,h,t,u,b,v)=(2,2,2,1,0,1,1) $ & optimal\\
2 & 4 & 15 & 15 & 1 & 15 & $ (l,s,h,t,u,b,v)=(2,2,2,1,0,2,1) $ & optimal\\
2 & 5 & 3 & 31 & 26 & 3 & $ (l,s,h,t,u,b,v)=(5,5,5,1,0,0,1) $ & optimal\\
2 & 6 & 3 & 63 & 57 & 3 & $ (l,s,h,t,u,b,v)=(2,2,2,1,0,0,1) $ & optimal\\
2 & 6 & 5 & 63 & 51 & 5 & $ (l,s,h,t,u,b,v)=(6,6,6,2,0,0,1) $ & optimal\\
2 & 6 & 9 & 63 & 39 & 9 & $ (l,s,h,t,u,b,v)=(2,2,6,1,0,0,3) $ & best known\\
2 & 6 & 21 & 63 & 18 & 21 & $ (l,s,h,t,u,b,v)=(2,2,6,1,0,0,7) $ & best known\\
2 & 6 & 27 & 63 & 10 & 27 & $ (l,s,h,t,u,b,v)=(3,3,6,1,0,0,9) $ & best known\\
2 & 6 & 31 & 63 & 7 & 31 & $ (l,s,h,t,u,b,v)=(2,2,2,1,0,3,1) $ & optimal\\
2 & 6 & 63 & 63 & 1 & 63 & $ (l,s,h,t,u,b,v)=(2,2,2,1,0,4,1) $ & optimal\\
2 & 7 & 3 & 127 & 120 & 3 & $ (l,s,h,t,u,b,v)=(7,7,7,1,0,0,1) $ & optimal\\
2 & 8 & 3 & 255 & 247 & 3 & $ (l,s,h,t,u,b,v)=(2,2,2,1,0,0,1) $ & optimal\\
2 & 8 & 5 & 255 & 239 & 5 & $ (l,s,h,t,u,b,v)=(4,4,4,2,0,0,1) $ & optimal\\
2 & 8 & 7 & 255 & 231 & 7 & $ (l,s,h,t,u,b,v)=(2,2,2,1,0,1,1) $ & best known\\
2 & 8 & 11 & 255 & 215 & 11 & $ (l,s,h,t,u,b,v)=(4,4,4,2,0,1,1) $ & best known\\
2 & 8 & 15 & 255 & 199 & 15 & $ (l,s,h,t,u,b,v)=(2,2,2,1,0,2,1) $ & best known\\
2 & 8 & 17 & 255 & 191 & 17 & $ (l,s,h,t,u,b,v)=(8,8,8,4,0,0,1) $ & best known\\
2 & 8 & 23 & 255 & 171 & 23 & $ (l,s,h,t,u,b,v)=(4,4,4,2,0,2,1) $ & best known\\
2 & 8 & 31 & 255 & 139 & 31 & $ (l,s,h,t,u,b,v)=(2,2,2,1,0,3,1) $ & best known\\
2 & 8 & 47 & 255 & 99 & 47 & $ (l,s,h,t,u,b,v)=(4,4,4,2,0,3,1) $ & best known\\
2 & 8 & 51 & 255 & 91 & 51 & $ (l,s,h,t,u,b,v)=(2,2,8,1,0,0,17) $ & best known\\
2 & 8 & 85 & 255 & 47 & 85 & $ (l,s,h,t,u,b,v)=(4,4,8,2,0,0,17) $ & best known\\
2 & 8 & 119 & 255 & 13 & 119 & $ (l,s,h,t,u,b,v)=(2,4,8,1,1,0,17) $ & best known\\
2 & 8 & 127 & 255 & 9 & 127 & $ (l,s,h,t,u,b,v)=(2,2,2,1,0,5,1) $ & optimal\\
2 & 8 & 255 & 255 & 1 & 255 & $ (l,s,h,t,u,b,v)=(2,2,2,1,0,6,1) $ & optimal\\
3 & 2 & 4 & 8 & 4 & 4 & $ (l,s,h,t,u,b,v)=(2,2,2,1,0,0,1) $ & optimal\\
3 & 3 & 4 & 26 & 20 & 4 & $ (l,s,h,t,u,b,v)=(3,3,3,1,0,0,1) $ & optimal\\
3 & 4 & 4 & 80 & 72 & 4 & $ (l,s,h,t,u,b,v)=(2,2,2,1,0,0,1) $ & optimal\\
3 & 4 & 8 & 80 & 60 & 8 & $ (l,s,h,t,u,b,v)=(2,2,4,1,0,0,2) $ & best known\\
3 & 4 & 10 & 80 & 56 & 10 & $ (l,s,h,t,u,b,v)=(4,4,4,2,0,0,1) $ & best known\\
3 & 4 & 14 & 80 & 46 & 14 & $ (l,s,h,t,u,b,v)=(2,2,2,1,0,1,1) $ & best known\\
3 & 4 & 40 & 80 & 16 & 40 & $ (l,s,h,t,u,b,v)=(2,2,4,1,0,0,10) $ & best known\\
3 & 4 & 44 & 80 & 11 & 44 & $ (l,s,h,t,u,b,v)=(2,2,2,1,0,2,1) $ & best known\\
3 & 5 & 4 & 242 & 232 & 4 & $ (l,s,h,t,u,b,v)=(5,5,5,1,0,0,1) $ & optimal\\
4 & 2 & 5 & 15 & 9 & 5 & $ (l,s,h,t,u,b,v)=(2,2,2,1,0,0,1) $ & optimal\\
4 & 3 & 5 & 63 & 54 & 5 & $ (l,s,h,t,u,b,v)=(3,3,3,1,0,0,1) $ & best known\\
4 & 4 & 5 & 255 & 243 & 5 & $ (l,s,h,t,u,b,v)=(2,2,2,1,0,0,1) $ & best known\\
\end{longtable}\endgroup

\subsection{ \texorpdfstring{Designed distance of the form 
$\delta=k\frac{p^a-1}{p^g-1}+1$
}{}}
\begin{theorem}
\label{theorem8}
Let $q=p^e$ for some prime $p$ and positive integer $e$, let  $a$ and $m$ be positive integers with $a\mid em$, and let $g=\gcd(a,e)$. 
 Suppose that $\delta\geq 2$ is an integer of the form 
\begin{equation}\notag
    \delta= k\frac{p^a-1}{p^g-1}+1,  
\end{equation}
where $k$ is an  integer satisfying $1\leq k \leq p^g-2$.  Then 
\begin{equation}\notag
d\bigl(\mathcal{C}_{(q,m,\delta)}\bigr)=\delta.
\end{equation} 
\end{theorem}
\begin{proof}
Since $a\mid em$, $g=\gcd(a,e)$ and $q=p^e$, we first  have
\begin{equation}\notag
\mathbb{F}_{p^a}\subseteq \mathbb{F}_{p^{em}}=\mathbb{F}_{q^m}
\quad \hbox{and}\quad 
\mathbb{F}_{p^g}\subseteq
\mathbb{F}_{p^{e}}=\mathbb{F}_{q}.
\end{equation}
For simplicity, let $v=\frac{p^a-1}{p^g-1}$. 
Choose a subset
$
\mathcal{Y}\subseteq \mathbb{F}_{p^g}^*
$ such that $
|\mathcal{Y}|=k,
$ 
and define
\begin{equation}\notag
\mathcal{A}
=
\{0\}\cup
\left\{
x\in \mathbb{F}_{p^a}^*: x^v\in \mathcal{Y}
\right\}.
\end{equation}
By Lemma \ref{lemmas4}, for each $\gamma\in \mathcal{Y}$, there exist  $v$ distinct solutions to $x^v=\gamma$ in $\mathbb{F}_{p^a}^{*}$, and the solution sets corresponding to distinct values of $\gamma$ are pairwise disjoint. 
 Therefore, 
\begin{equation}\notag
   |\mathcal{A}|=1+kv= \delta,
\end{equation}
and  the vanishing polynomial of $\mathcal{A}$ is given by 
\begin{equation}\notag
P_{\mathcal{A}}(x)
=
x\prod_{\gamma\in \mathcal{Y}}(x^{v}-\gamma).
\end{equation}

We next show that $\mathcal{A}$ is a strict $\mathbb{F}_q$-good zero-set.  
A direct calculation gives
\begin{equation}\notag
P_{\mathcal{A}}'(0)
=
\prod_{\gamma\in \mathcal{Y}}(-\gamma)\in \mathbb{F}_{p^g}^{*}. 
\end{equation}
For each  $\alpha\in \mathcal{A}\setminus\{0\}$, we have $\alpha^v\in \mathcal{Y} \subseteq \mathbb{F}_{p^g}^{*}$. Therefore, we obtain 
\begin{equation}\notag
P_{\mathcal{A}}'(\alpha)
=
v \alpha^{v}
\prod_{\substack{\gamma\in \mathcal{Y}\\ \gamma \neq \alpha^v}} 
(\alpha^v-\gamma)\in  v\mathbb{F}_{p^g}^{*}.
\end{equation}
Noting  that 
$p\nmid v$ and recalling  $\mathbb{F}_{p^g}\subseteq \mathbb{F}_{q}$, we obtain  $
    v\mathbb{F}_{p^g}^*\subseteq \mathbb{F}_{p^g}^*\subseteq \mathbb{F}_{q}^*.
$ 
Therefore, we can conclude that 
\begin{equation}\notag
   P_{\mathcal{A}}'(\alpha)\in \mathbb{F}_q^* \quad \hbox{for all }\alpha\in \mathcal{A}.  
\end{equation}
It follows that  $\mathcal{A}$ is a strict $\mathbb{F}_q$-good zero-set. 

We now construct an $\mathbb{F}_q$-good zero-set from $\mathcal{A}$. 
Since
$|\mathbb{F}_{p^g}^*|=p^g-1
$ and $
k\leq p^g-2$,
there exists 
$
\eta\in \mathbb{F}_{p^g}^*\setminus \mathcal{Y}.
$ 
By Lemma \ref{lemmas4} again, there exists 
 $\beta\in \mathbb{F}_{p^a}^*$ such
that
$
\beta^v=\eta.
$ 
Since $\eta\notin \mathcal{Y}$, we have $\beta\in \mathbb{F}_{p^a}^*\setminus \mathcal{A}$. 
Furthermore, 
for any $\alpha\in \mathcal{A}$, we have $\alpha-\beta\in \mathbb{F}_{p^a}^*$ and hence 
\begin{equation*}
    ((\alpha-\beta)^{v})^{p^g-1}=(\alpha-\beta)^{p^a-1}=1.
\end{equation*}
This implies $(\alpha-\beta)^v\in \mathbb{F}_{p^g}^*$. 
Consequently, 
\begin{equation}\notag
\begin{aligned}
(\alpha-\beta)^{|\mathcal{A}|-1}
&=
\left((\alpha-\beta)^{v}\right)^k  \in \mathbb{F}_{p^g}^*\subseteq\mathbb{F}_q^*.
\end{aligned}
\end{equation}
Applying Theorem~\ref{theorem4}, we obtain an $\mathbb{F}_q$-good zero-set 
$$\mathcal{A}^{+}=\{0\}\cup \left\{\frac{1}{\alpha-\beta}:\alpha\in \mathcal{A}\right\}$$ satisfying 
\begin{equation}\notag
    |\mathcal{A}^{+}|=|\mathcal{A}|+1=\delta+1.
\end{equation}
Note  that $\mathcal{A}^{+}\subseteq \mathbb{F}_{p^a}\subseteq \mathbb{F}_{q^m}$. Therefore, by 
applying Theorem~\ref{theorem1}, we obtain 
$d\bigl(\mathcal{C}_{(q,m,\delta)}\bigr)=\delta$. This completes the proof.
\end{proof}

Combining Theorem \ref{theorem8} with the dimension formula given in
\cite[Theorem 4.10]{CherchemEtAl2020}, we obtain the following corollary,
which extends \cite[Theorem 7]{Chen2026} by removing the restriction
$k\mid(q-1)$.
\begin{corollary}\label{corollary:theorem8}
Let $q$ be a prime power, and let $m$ be a positive integer.
Let
\begin{equation}\notag
\delta=k\frac{q^m-1}{q-1}+1,
\end{equation}
where $k$ is an integer satisfying $1\leq k\leq q-2$.
Then $\mathcal{C}_{(q,m,\delta)}$ has minimum distance $\delta$
and dimension $(q-k)^m-1$.
\end{corollary}
\begin{proof}
Let $q=p^e$, where $p$ is a prime and $e$ is a positive integer.
Taking $a=em$ in Theorem~\ref{theorem8}, we have
$g=\gcd(em,e)=e$ and
\begin{equation}\notag
k\frac{p^a-1}{p^g-1}+1
=
k\frac{q^m-1}{q-1}+1
=
\delta.
\end{equation}
Thus, Theorem~\ref{theorem8} gives
$d(\mathcal{C}_{(q,m,\delta)})=\delta$.
The dimension follows from
\cite[Theorem~4.10]{CherchemEtAl2020}.
\end{proof}

\begingroup\scriptsize\setlength{\tabcolsep}{4pt}\renewcommand{\arraystretch}{0.85}
\begin{longtable}[c]{@{}rrrrrrcc@{}}
\caption{Examples of BCH codes in Theorem 8}\label{tab:examples-theorem-8}\\
\toprule
\multicolumn{1}{@{}c}{$q$} & \multicolumn{1}{c}{$m$} & \multicolumn{1}{c}{$\delta$} & \multicolumn{1}{c}{$n$} & \multicolumn{1}{c}{$\dim$} & \multicolumn{1}{c}{$d$} & \multicolumn{1}{c}{Parameter choices} & \multicolumn{1}{c@{}}{optimality}\\
\midrule
\endfirsthead
\multicolumn{8}{c}{\tablename\ \thetable\ -- continued from previous page}\\
\toprule
\multicolumn{1}{@{}c}{$q$} & \multicolumn{1}{c}{$m$} & \multicolumn{1}{c}{$\delta$} & \multicolumn{1}{c}{$n$} & \multicolumn{1}{c}{$\dim$} & \multicolumn{1}{c}{$d$} & \multicolumn{1}{c}{Parameter choices} & \multicolumn{1}{c@{}}{optimality}\\
\midrule
\endhead
\midrule
\multicolumn{8}{r}{Continued on next page}\\
\endfoot
\bottomrule
\endlastfoot
3 & 1 & 2 & 2 & 1 & 2 & $ (a,g,k)=(1,1,1) $ & optimal\\
3 & 2 & 2 & 8 & 6 & 2 & $ (a,g,k)=(1,1,1) $ & optimal\\
3 & 2 & 5 & 8 & 3 & 5 & $ (a,g,k)=(2,1,1) $ & optimal\\
3 & 3 & 2 & 26 & 23 & 2 & $ (a,g,k)=(1,1,1) $ & optimal\\
3 & 3 & 14 & 26 & 7 & 14 & $ (a,g,k)=(3,1,1) $ & optimal\\
3 & 4 & 2 & 80 & 76 & 2 & $ (a,g,k)=(1,1,1) $ & optimal\\
3 & 4 & 41 & 80 & 15 & 41 & $ (a,g,k)=(4,1,1) $ & best known\\
3 & 5 & 2 & 242 & 237 & 2 & $ (a,g,k)=(1,1,1) $ & optimal\\
3 & 5 & 122 & 242 & 31 & 122 & $ (a,g,k)=(5,1,1) $ & best known\\
4 & 1 & 2 & 3 & 2 & 2 & $ (a,g,k)=(2,2,1) $ & optimal\\
4 & 1 & 3 & 3 & 1 & 3 & $ (a,g,k)=(2,2,2) $ & optimal\\
4 & 2 & 2 & 15 & 13 & 2 & $ (a,g,k)=(2,2,1) $ & optimal\\
4 & 2 & 6 & 15 & 8 & 6 & $ (a,g,k)=(4,2,1) $ & optimal\\
4 & 2 & 11 & 15 & 3 & 11 & $ (a,g,k)=(4,2,2) $ & optimal\\
4 & 3 & 2 & 63 & 60 & 2 & $ (a,g,k)=(2,2,1) $ & optimal\\
4 & 3 & 22 & 63 & 26 & 22 & $ (a,g,k)=(6,2,1) $ & best known\\
4 & 3 & 43 & 63 & 7 & 43 & $ (a,g,k)=(6,2,2) $ & optimal\\
4 & 4 & 2 & 255 & 251 & 2 & $ (a,g,k)=(2,2,1) $ & optimal\\
4 & 4 & 171 & 255 & 15 & 171 & $ (a,g,k)=(8,2,2) $ & best known\\
\end{longtable}\endgroup

\subsection{Designed distances of the form
\texorpdfstring{$q^{2r}+1$ and $q^{3r}+1$}
{q to the power 2r plus 1 and q to the power 3r plus 1}}

\begin{theorem}\label{theorem9}
Let $q=p^e$ for some prime $p$ and positive integer $e$, and let $r$ be a positive integer. 
\begin{enumerate}
\item[]
If $5r\mid m$, then
\begin{equation}\label{thm10e1}
d\bigl(\mathcal{C}_{(q,m,q^{2r}+1)}\bigr) = q^{2r}+1.
\end{equation}

\item[]
If $8r\mid m$, then
\begin{equation}\label{thm10e2}
d\bigl(\mathcal{C}_{(q,m,q^{3r}+1)}\bigr) = q^{3r}+1.
\end{equation}
\end{enumerate}
\end{theorem}

\begin{proof}
We first prove \eqref{thm10e1}. Let
\begin{equation}\notag
H(x)=x^{\frac{q^{2r}}{p}}+x^{\frac{q^{2r}-q^r}{p}}
\end{equation}
and
\begin{equation}\notag
\begin{aligned}
U(x)
&=-1+xH(x)^p\\
&=x^{q^{2r}+1}+x^{q^{2r}-q^r+1}-1.
\end{aligned}
\end{equation}
Then $U(x)$ has the form
required in Theorem \ref{theorem5}  with $\beta=-1$.

We now show that $U(x)$ splits completely over
$\mathbb{F}_{q^{5r}}$, that is, all its roots lie in
$\mathbb{F}_{q^{5r}}$.
Let $\alpha$ be a root
of $U(x)$. Noting that $U(0)=-1$, we have $\alpha\neq0$. Moreover,
\begin{equation}\notag
\alpha^{q^{2r}+1}+\alpha^{q^{2r}-q^r+1}=1.
\end{equation}
Dividing both sides by $\alpha^{q^{2r}+1}$ gives
\begin{equation}\label{bz}
1+\alpha^{-q^r}=\alpha^{-q^{2r}-1}.
\end{equation}
For each integer $i\geq0$, define
$
x_i=\alpha^{-q^{ri}}.
$ 
Then \eqref{bz} becomes
\begin{equation}\notag
1+x_1=x_0x_2
\end{equation}
Repeatedly applying the map $x\mapsto x^{q^r}$ to both sides yields
\begin{equation}\notag
1+x_{i+1}=x_ix_{i+2}
\qquad\text{for all }i\geq0.
\end{equation}
A direct calculation then gives
\begin{equation}\notag
x_2=\frac{x_1+1}{x_0},\quad
x_3=\frac{x_0+x_1+1}{x_0x_1},\quad
x_4=\frac{x_0+1}{x_1},\quad
x_5=x_0.
\end{equation}
The last equality is equivalent to 
$\alpha^{-q^{5r}}=\alpha^{-1}$, and hence
$\alpha^{q^{5r}}=\alpha$.
Thus, we have $\alpha\in\mathbb{F}_{q^{5r}}$.
Therefore, $U(x)$ splits completely over $\mathbb{F}_{q^{5r}}$.

By Theorem \ref{theorem5}, the set
\begin{equation}\notag
\mathcal{A}
=
\{0\}\cup
\{\alpha\in\mathbb{F}_{q^{5r}}:U(\alpha)=0\}
\end{equation}
is an $\mathbb{F}_q$-good zero-set satisfying
\begin{equation}\notag
|\mathcal{A}|=\deg(U)+1=q^{2r}+2.
\end{equation}
Since $5r\mid m$, we have
$\mathbb{F}_{q^{5r}}\subseteq\mathbb{F}_{q^m}$,
and hence $\mathcal{A}\subseteq\mathbb{F}_{q^m}$.
Then applying Theorem~\ref{theorem1}, we obtain \eqref{thm10e1}.

We next prove \eqref{thm10e2} under the assumption $8r\mid m$.
Let
\begin{equation}\notag
\widetilde{H}(x)
= x^{\frac{q^{3r}}{p}}
+ x^{\frac{q^{3r}-q^r}{p}}
+ x^{\frac{q^{3r}-q^{2r}}{p}},
\end{equation}
and
\begin{equation}\notag
\begin{aligned}
\widetilde{U}(x)
&= -1+x\widetilde{H}(x)^p \\
&= x^{q^{3r}+1}
+ x^{q^{3r}-q^r+1}
+ x^{q^{3r}-q^{2r}+1}-1.
\end{aligned}
\end{equation}
We now show that $\widetilde{U}(x)$ splits completely over
$\mathbb{F}_{q^{8r}}$, that is, all its roots lie in
$\mathbb{F}_{q^{8r}}$.
Let $\alpha$ be a root of
$\widetilde{U}(x)$. Then $\alpha\neq0$ and
\begin{equation}\notag
\alpha^{q^{3r}+1}
+\alpha^{q^{3r}-q^r+1}
+\alpha^{q^{3r}-q^{2r}+1}=1.
\end{equation}
Dividing both sides by $\alpha^{q^{3r}+1}$ yields
\begin{equation}\notag
1+\alpha^{-q^r}+\alpha^{-q^{2r}}
=\alpha^{-q^{3r}-1}.
\end{equation}
For each integer $i\geq0$, define $y_i=\alpha^{-q^{ri}}$.
Then the preceding equality becomes
\begin{equation}\notag
1+y_1+y_2=y_0y_3.
\end{equation}
Repeatedly applying the map $y\mapsto y^{q^r}$ to both sides yields
\begin{equation}\label{eq:eight-recurrence}
1+y_{i+1}+y_{i+2}=y_iy_{i+3}
\qquad\text{for all }i\geq0.
\end{equation}
Successively expressing $y_3,\ldots,y_8$ in terms of
$y_0$, $y_1$, and $y_2$ using \eqref{eq:eight-recurrence}, we obtain
\begin{equation}\notag
y_8=y_0.
\end{equation}
Equivalently,
\begin{equation}\notag
\alpha^{-q^{8r}}=\alpha^{-1}.
\end{equation}
Hence $\alpha^{q^{8r}}=\alpha$, so
$\alpha\in\mathbb{F}_{q^{8r}}$, proving that
$\widetilde{U}(x)$ splits completely over $\mathbb{F}_{q^{8r}}$.

By Theorem~\ref{theorem5}, the set
\begin{equation}\notag
\widetilde{\mathcal{A}}
=\{0\}\cup
\{\alpha\in\mathbb{F}_{q^{8r}}:\widetilde{U}(\alpha)=0\}
\end{equation}
is an $\mathbb{F}_q$-good zero-set with
\begin{equation}\notag
|\widetilde{\mathcal{A}}|
=\deg(\widetilde{U})+1=q^{3r}+2.
\end{equation}
Since $8r\mid m$, we have
$\mathbb{F}_{q^{8r}}\subseteq\mathbb{F}_{q^m}$, and thus
$\widetilde{\mathcal{A}}\subseteq\mathbb{F}_{q^m}$.
Applying Theorem~\ref{theorem1}, we obtain \eqref{thm10e2}.
This completes the proof.
\end{proof}

\begingroup\scriptsize\setlength{\tabcolsep}{4pt}\renewcommand{\arraystretch}{0.85}
\begin{longtable}[c]{@{}rrrrrrc@{}}
\caption{Examples of BCH codes in Theorem 9}\label{tab:examples-theorem-9}\\
\toprule
\multicolumn{1}{@{}c}{$q$} & \multicolumn{1}{c}{$m$} & \multicolumn{1}{c}{$\delta$} & \multicolumn{1}{c}{$n$} & \multicolumn{1}{c}{$\dim$} & \multicolumn{1}{c}{$d$} & \multicolumn{1}{c@{}}{optimality}\\
\midrule
\endfirsthead
\multicolumn{7}{c}{\tablename\ \thetable\ -- continued from previous page}\\
\toprule
\multicolumn{1}{@{}c}{$q$} & \multicolumn{1}{c}{$m$} & \multicolumn{1}{c}{$\delta$} & \multicolumn{1}{c}{$n$} & \multicolumn{1}{c}{$\dim$} & \multicolumn{1}{c}{$d$} & \multicolumn{1}{c@{}}{optimality}\\
\midrule
\endhead
\midrule
\multicolumn{7}{r}{Continued on next page}\\
\endfoot
\bottomrule
\endlastfoot
2 & 5 & 5 & 31 & 21 & 5 & optimal\\
2 & 8 & 9 & 255 & 223 & 9 & best known\\
3 & 5 & 10 & 242 & 212 & 10 & best known\\
\end{longtable}\endgroup

\section{Conclusion}
We introduced $\mathbb{F}_q$-good zero-sets and reduced the question
of whether $d(\mathcal{C}_{(q,m,\delta)})=\delta$ holds to the
existence of such a set of cardinality $\delta+1$ in
$\mathbb{F}_{q^m}$.
We developed several methods for constructing $\mathbb{F}_q$-good
zero-sets. Using these methods, together with  $\mathbb{F}_q$-good zero-sets arising
from known minimum-distance results, we obtained families of
primitive narrow-sense BCH codes attaining the following designed
distances:
\begin{itemize}
    \item $\delta=\bigl(1+v(tp^u-1)\bigr)p^b-1$,
    under the conditions in Theorem~\ref{theorem6};

    \item $\delta=\bigl(1+v((q^t+2)p^u-1)\bigr)p^b-1$,
    under the conditions in Theorem~\ref{theorem7};

    \item $\delta=k\frac{p^a-1}{p^g-1}+1$, where $g=\gcd(a,e)$,
    under the conditions in Theorem~\ref{theorem8};

    \item $\delta=q^{2r}+1$ for positive integers $r$ with $5r\mid m$,
    as established in Theorem~\ref{theorem9};

    \item $\delta=q^{3r}+1$ for positive integers $r$ with $8r\mid m$,
    as established in Theorem~\ref{theorem9}.
\end{itemize}
These families include several known results as special cases.
Our methods for constructing new good zero-sets from known ones may
offer a  route to further families of BCH codes attaining
their designed distances.

\section{Acknowledgments}
The author gratefully acknowledges Professor Maosheng Xiong
and Professor Cunsheng Ding for their financial support
and helpful comments. The author used ChatGPT to assist with literature searches, exploratory mathematical discussions, and language editing. The mathematical results, arguments, and proofs presented in this paper were developed and independently verified by the author, who takes full responsibility for the content of the work.

\appendix 
\begin{proof}[\normalfont\bfseries Proof of Lemma~\ref{adlemma1}] 
Let $j$ be an integer with $0\leq j\leq |\mathcal{D}|-2$.
For each $\beta\in \mathcal{D}$, define the polynomial
    \begin{equation}\notag
        L_{\beta}(x) =  \frac{\prod\limits_{\alpha\in \mathcal{D}\setminus \{\beta\}}(x-\alpha)}{\prod\limits_{\alpha\in \mathcal{D}\setminus\{ \beta\}}(\beta-\alpha)}\beta^j. 
    \end{equation}
    For each $\beta, \gamma\in \mathcal{D}$,
    it can be easily verified that 
    \begin{equation*}
       L_{\beta}(\gamma)=\begin{cases}
           0& \hbox{if } \gamma\neq \beta, \\
           \beta^j&\hbox{if }\gamma=\beta. 
       \end{cases}
\end{equation*}
Thus, we have 
$$\sum_{\beta\in \mathcal{D}}L_{\beta}(\gamma)-\gamma^j=0\quad \hbox{for all } \gamma\in \mathcal{D}.$$ 
Note that the degree of the polynomial $\sum_{\beta\in \mathcal{D}}L_{\beta}(x)-x^j$ is at most $|\mathcal{D}|-1$.  It follows that  $\sum_{\beta\in \mathcal{D}}L_{\beta}(x)-x^j$ is the  zero polynomial, and hence 
\begin{equation}\label{equation2}
        \sum_{\beta\in \mathcal{D}}L_{\beta}(x) = x^j. 
    \end{equation} 
Recall that $P_{\mathcal{D}}'(\beta)=\prod_{\alpha\in \mathcal{D}\setminus\{\beta\}}(\beta-\alpha)$ for each $\beta\in \mathcal{D}$. The coefficient of $x^{|\mathcal{D}|-1}$ on the left-hand side of \eqref{equation2} is exactly $\sum_{\beta\in \mathcal{D}}\frac{\beta^j}{P_{\mathcal{D}}'(\beta)}$. On the right-hand side, since $j \le |\mathcal{D}|-2$, the coefficient of $x^{|\mathcal{D}|-1}$ in $x^j$ is $0$. Comparing the coefficients of $x^{|\mathcal{D}|-1}$ on both sides yields \eqref{lemma1e1}.
\end{proof}

\hspace*{\fill}

\noindent\textbf{Proof of Lemma \ref{lemma2}. }
   Consider the $\mathbb{F}_p$-linear map $\phi: \mathbb{F}\to \mathbb{F}$ defined as 
   $
\phi(y)=y^p-y.
$ 
We first show that \begin{equation}\label{lema3e1}
    \operatorname{Im}(\phi)=
\ker\bigl(\operatorname{Tr}_{\mathbb{F}/\mathbb{F}_p}\bigr).
\end{equation} 
For any   $y\in \mathbb{F}$, we have 
\begin{equation}\notag \Tr_{\mathbb{F}/\mathbb{F}_{p}}(\phi(y)) = \Tr_{\mathbb{F}/\mathbb{F}_{p}}(y^p)-\Tr_{\mathbb{F}/\mathbb{F}_{p}}(y)=0.
\end{equation}
This implies that $\operatorname{Im}(\phi)\subseteq\ker(\Tr_{\mathbb{F}/\mathbb{F}_{p}})$.

Note that $\phi(y)=0$ if and only if $y\in\mathbb{F}_p$. Therefore,
$\ker(\phi)=\mathbb{F}_p$. By the rank--nullity theorem,
\begin{equation}\notag
\dim_{\mathbb{F}_p}\operatorname{Im}(\phi)
=\dim_{\mathbb{F}_p}(\mathbb{F})-\dim_{\mathbb{F}_p}\ker(\phi)= \dim_{\mathbb{F}_p}(\mathbb{F})-\dim_{\mathbb{F}_p}(\mathbb{F}_p).
\end{equation}
Moreover, the trace map $\Tr_{\mathbb{F}/\mathbb{F}_{p}}$ is an $\mathbb{F}_{p}$-linear map satisfying  $\operatorname{Im}(\Tr_{\mathbb{F}/\mathbb{F}_{p}})=\mathbb{F}_{p}$. Therefore,
\begin{equation}\notag   \dim_{\mathbb{F}_p}\ker(\Tr_{\mathbb{F}/\mathbb{F}_{p}})=\dim_{\mathbb{F}_p}(\mathbb{F})-\dim_{\mathbb{F}_p}\operatorname{Im}(\Tr_{\mathbb{F}/\mathbb{F}_{p}})=
\dim_{\mathbb{F}_p}(\mathbb{F})-\dim_{\mathbb{F}_p}(\mathbb{F}_p).
\end{equation}
 Combining  $\operatorname{Im}(\phi)\subseteq\ker(\Tr_{\mathbb{F}/\mathbb{F}_{p}})$ obtained above with these equalities, we obtain \eqref{lema3e1}.

Since $\Tr_{\mathbb{F}/\mathbb{F}_{p}}(\alpha)=0$, we have $\alpha\in \ker(\Tr_{\mathbb{F}/\mathbb{F}_{p}})$. By \eqref{lema3e1}, it follows that  $\alpha\in \operatorname{Im}(\phi)$. Therefore, there exists $\gamma\in \mathbb{F}$ such that 
$\gamma^p-\gamma=\alpha$.   Then 
\begin{equation}\notag
    (\gamma+\beta)^p-(\gamma+\beta)-\alpha
    =
    \gamma^p-\gamma-\alpha
    +
    (\beta^p-\beta)
    =0\quad \hbox{for all }\beta\in\mathbb{F}_p.
\end{equation}
Consequently,  $\gamma+\beta$ is a
 root of $y^p-y-\alpha$ for every  $\beta\in \mathbb{F}_p$. Since these constitute $p$ distinct roots and the polynomial $y^p-y-\alpha$ has degree $p$, it splits completely over $\mathbb{F}$.  \qed

\hspace*{\fill}

\noindent\textbf{Proof of Lemma \ref{lemmas4}. }
For simplicity, let $L=(q^h-1)/(q^s-1)$. Since $v\mid L$ and $L\mid(q^h-1)$, we have $v\mid(q^h-1)$.
Since $q$ is a power of $p$, it follows that
$p\nmid v$.
Let $\theta$ be a primitive element of $\mathbb{F}_{q^h}$. Then $\theta$ has multiplicative order $q^h - 1$, and hence  
$\theta^L$ has multiplicative order $q^s-1$. Consequently, $\theta^L$ generates the multiplicative group
$\mathbb{F}_{q^s}^{*}$; that is,
\begin{equation}\notag
\mathbb{F}_{q^s}^{*}=\langle\theta^L\rangle.
\end{equation}
Thus, for any $\alpha\in\mathbb{F}_{q^s}^{*}$, we can write
$\alpha=\theta^{Lj}$ for some integer $j$.
Let 
$
\beta=\theta^{(L/v)j},
$ and 
$\eta=\theta^{(q^h-1)/v}$.
Then we have $\beta^v=\alpha$, and $\eta$ has multiplicative order $v$.
Consequently, the $v$ distinct elements
$\beta,\beta\eta,\ldots,\beta\eta^{v-1}$ are nonzero solutions to
$x^v=\alpha$.

Finally, suppose that \eqref{dis} fails. Then there exists
$\gamma\in\mathbb{F}_{q^h}$ such that
$\alpha=\gamma^v=\alpha'$, contradicting $\alpha\ne\alpha'$.
Thus \eqref{dis} holds, completing the proof.
\qed

\hspace*{\fill}

\noindent\textbf{Proof of Lemma \ref{lemma5}. }
Since $1\leq t\leq q$, there exists
an $\mathbb{F}_p$-linear subspace
$\mathcal{U}\subseteq\mathbb{F}_q$ such that
$\dim_{\F_p}\mathcal{U}=\lceil\log_p t\rceil$. Then
$|\mathcal{U}|\geq t$, so we can choose a set
$\mathcal{A}\subseteq\mathcal{U}$ such that $|\mathcal{A}|=t$ and $0\in \mathcal{A}$. It follows that
$\Span_{\mathbb{F}_p}(\mathcal{A})\subseteq\mathcal{U}$, and hence
\begin{equation}\notag
\dim_{\F_p}\Span_{\F_p}(\mathcal{A})
\leq \dim_{\F_p}\mathcal{U}
= \lceil\log_p t\rceil.
\end{equation} 
Moreover, 
since $\mathcal{A}\subseteq\mathcal{U}\subseteq\mathbb{F}_q$, we also have
\begin{equation}\notag
P_{\mathcal{A}}'(\alpha)
=
\prod_{\substack{\beta\in \mathcal{A}\\ \beta\neq \alpha}}(\alpha-\beta)
\in\mathbb{F}_q^* \quad \hbox{for all }\alpha\in \mathcal{A}.
\end{equation}
It follows that $\mathcal{A}$ is a strict
$\mathbb{F}_q$-good zero-set. This completes the proof.
\qed

\end{document}